\documentclass[a4paper,twocolumn,11pt]{quantumarticle}
\pdfoutput=1
\usepackage[utf8]{inputenc}
\usepackage[english]{babel}
\usepackage[T1]{fontenc}
\usepackage{amsmath,amssymb,amsthm}
\usepackage{array}
\usepackage{physics}
\usepackage{mathtools}
\usepackage[dvipsnames]{xcolor}
\usepackage{float}
\usepackage{hyperref}
\usepackage{cleveref}
\usepackage{enumitem}
\usepackage{listings}
\usepackage{tcolorbox}
\usepackage{booktabs}
\usepackage{quantikz}
\usetikzlibrary{tikzmark, arrows.meta}
\makeatletter\let\switch@array\relax\makeatother

\definecolor{linknavy}{rgb}{0.0, 0.2, 0.5}
\definecolor{leangray}{rgb}{0.97, 0.97, 0.97}
\definecolor{leankw}{rgb}{0.0, 0.0, 0.55}
\definecolor{leancomment}{rgb}{0.0, 0.5, 0.0}
\hypersetup{colorlinks=true, linkcolor=linknavy, citecolor=linknavy, urlcolor=linknavy}

\lstdefinelanguage{Lean4}{
keywords={theorem, lemma, def, noncomputable, intro, exact, fun, by, rw, simp, ring, congr, refine, obtain, have, show, unfold, apply, decide, namespace, end, open, import, abbrev, if, then, else, with}, keywordstyle=\color{leankw}\bfseries, comment=[l]{--}, commentstyle=\color{leancomment}\itshape, morestring=[b]", literate= {ℝ}{{$\mathbb{R}$}}1 {ℂ}{{$\mathbb{C}$}}1 {ℕ}{{$\mathbb{N}$}}1 {α}{{$\alpha$}}1 {β}{{$\beta$}}1 {η}{{$\eta$}}1 {θ}{{$\theta$}}1 {π}{{$\pi$}}1 {ψ}{{$\psi$}}1 {ℓ}{{$\ell$}}1 {≤}{{$\leq$}}1 {≥}{{$\geq$}}1 {≠}{{$\neq$}}1 {∀}{{$\forall$}}1 {∃}{{$\exists$}}1 {∑}{{$\sum$}}1 {∏}{{$\prod$}}1 {→}{{$\rightarrow$}}1 {↦}{{$\mapsto$}}1 {⟨}{{$\langle$}}1 {⟩}{{$\rangle$}}1 {∧}{{$\wedge$}}1 {∨}{{$\vee$}}1 {¬}{{$\neg$}}1 {σ}{{$\sigma$}}1 {∈}{{$\in$}}1 {₀}{{$_{0}$}}1 {ᶠ}{{$^{f}$}}1, }

\theoremstyle{plain}
\newtheorem{theorem}{Theorem}
\newtheorem{lemma}[theorem]{Lemma}
\newtheorem{proposition}[theorem]{Proposition}
\newtheorem{corollary}[theorem]{Corollary}
\newtheorem*{mainresult}{Main result}

\theoremstyle{definition}
\newtheorem{definition}[theorem]{Definition}
\newtheorem{remark}[theorem]{Remark}
\newtheorem{construction}[theorem]{Construction}
\newtheorem{assumption}[theorem]{Assumption}

\newcommand{\softmax}{\operatorname{softmax}}
\newcommand{\Real}{\operatorname{Re}}
\newcommand{\cSWAP}{\operatorname{CSWAP}}
\newcommand{\HilbertA}{\mathcal{H}_A}
\newcommand{\HilbertB}{\mathcal{H}_B}

\newcommand{\HilbertS}{\mathcal{H}_S}

\newsavebox{\hdlScore}
\newsavebox{\hdlTemp}
\newsavebox{\hdlSoft}
\newsavebox{\hdlVal}
\newsavebox{\hdlGate}

\title{A Quantum Roadmap for Softmax Attention: Exact Born-Rule Analogs for Softmax Attention on the Probability Simplex}
\author{Eric A. F. Reinhardt}
\affiliation{Department of Physics and Astronomy, University of Alabama, Tuscaloosa, AL 35487, USA}
\email{eareinhardt@crimson.ua.edu}
\author{Adam J. Hauser}
\affiliation{Department of Physics and Astronomy, University of Alabama, Tuscaloosa, AL 35487, USA}

\begin{document}

\sbox{\hdlScore}{
\begin{quantikz}[column sep=0.16cm, row sep=0.13cm]
\lstick{\scriptsize$|0\rangle_h$} & \gate{H} & \ctrl{1} & \gate{H} & \meter{} \\
\lstick{\scriptsize$|0\rangle_B$} & & \gate{\tilde V_Q / \tilde V_K} & &
\end{quantikz}}

\sbox{\hdlTemp}{
\begin{tabular}{@{}c@{}}
\begin{quantikz}[column sep=0.16cm, row sep=0.13cm]
\lstick{\scriptsize$|i,j\rangle$} & \ctrl{1} & \\
\lstick{\scriptsize$|0\rangle_{h^{(\ell)}}$} & \gate{R_y(\theta^{(\ell)}_{ij})} & \meter{|0\rangle}
\end{quantikz} \\[-1pt]
{\scriptsize $L$ post-selected rounds}
\end{tabular}}

\sbox{\hdlSoft}{
\begin{quantikz}[column sep=0.16cm, row sep=0.13cm]
\lstick{\scriptsize$|i\rangle_A$} & \ctrl{1} & \\
\lstick{\scriptsize$|0\rangle_S$} & \gate{U_{\mathrm{attn}}(\boldsymbol\theta_i)} & \meter{}
\end{quantikz}}

\sbox{\hdlVal}{
\begin{quantikz}[column sep=0.16cm, row sep=0.13cm]
\lstick{\scriptsize$|j\rangle_S$} & \ctrl{1} & & \trash{} \\
\lstick{\scriptsize$|0\rangle_{B'}$} & \gate{U_{\mathrm{enc}}} & \ctrl{1} & \trash{} \\
\lstick{\scriptsize$|0\rangle_{B_V}$} & & \gate{U_V} &
\end{quantikz}}

\sbox{\hdlGate}{
\begin{quantikz}[column sep=0.16cm, row sep=0.13cm]
\lstick{\scriptsize$|0\rangle_C$} & \gate{R_y(\eta)} & \ctrl{1} & \gate{H} & \meter{|0\rangle} \\
\lstick{\scriptsize$|X\rangle$} & & \gate{U_{\mathrm{full}}} & &
\end{quantikz}}

\maketitle

\begin{abstract}
The attention mechanism forms the foundation of many modern AI models such as the Transformer. In one subclass of problems where attention is used, inputs and outputs are bound to the probability simplex so that all outputs sum to one. In this setting, softmax attention admits an exact, component-by-component quantum realization. Attention scores are Hadamard-test statistics on block-encoded projections of amplitude-encoded inputs. The exponential softmax is the interior of a cosine-squared family generated by Born-rule measurement under an exact bijection, whose boundary expresses sparse attention with exact zeros at finite parameter values. The softmax temperature is a repetition count where post-selected measurement rounds realize discretized inverse temperature exactly. Value aggregation is a deterministic column-loading channel that dilates the column-stochastic value matrix. The gated residual is the preparation angle of a single ancilla, with the additive identity at a mixing angle of $\pi/2$. Every learnable parameter is a rotation-gate angle. The composed layer is exact in the infinite-shot limit with one measure-and-reload step per attention score; a fully-coherent variant is $\varepsilon$-approximate via quantum singular value transformation in the infinite depth limit. The algebraic core is machine-checked in Lean 4.
\end{abstract}

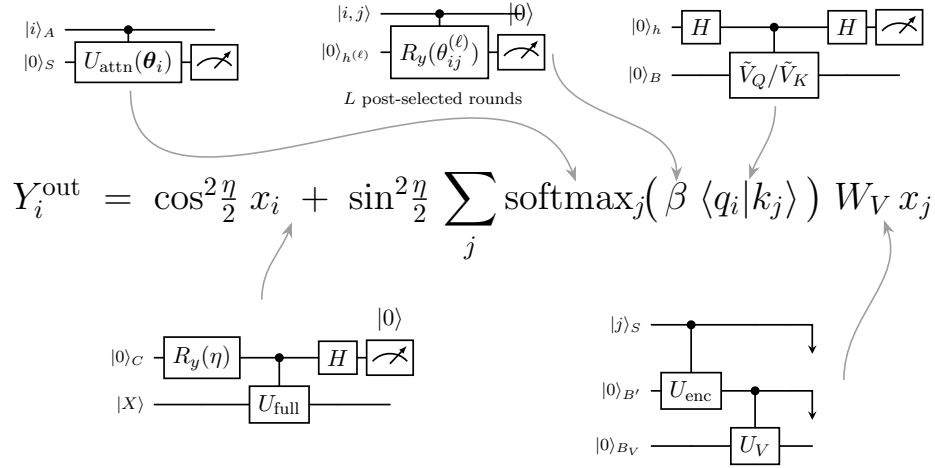
\begin{figure*}[t]
\begin{center}
\tikzmarknode{cSoft}{\scalebox{0.8}{\usebox{\hdlSoft}}}\hspace{0.9cm}
\tikzmarknode{cTemp}{\scalebox{0.8}{\usebox{\hdlTemp}}}\hspace{0.9cm}
\tikzmarknode{cScore}{\scalebox{0.8}{\usebox{\hdlScore}}}

\vspace{0.7cm}
{\Large $Y_i^{\mathrm{out}} \;=\; \tikzmarknode{fGate}{\cos^2\!\tfrac{\eta}{2}\; x_i \;+\; \sin^2\!\tfrac{\eta}{2}}\, \displaystyle\sum_j \tikzmarknode{fSoft}{\softmax_j}\!\big(\, \tikzmarknode{fTemp}{\beta}\; \tikzmarknode{fScore}{\langle q_i | k_j \rangle}\, \big)\; \tikzmarknode{fVal}{W_V\, x_j}$}

\vspace{0.3cm}
\tikzmarknode{cGate}{\scalebox{0.8}{\usebox{\hdlGate}}}\hspace{2.2cm}
\tikzmarknode{cVal}{\scalebox{0.8}{\usebox{\hdlVal}}}
\end{center}
\begin{tikzpicture}[overlay, remember picture]
\draw[-{Stealth}, gray!75, semithick] (cSoft.south) to[out=-90, in=110] (fSoft.north);
\draw[-{Stealth}, gray!75, semithick] (cTemp.east) to[out=-90, in=90] (fTemp.north);
\draw[-{Stealth}, gray!75, semithick] (cScore.south) to[out=-90, in=70] (fScore.north);
\draw[-{Stealth}, gray!75, semithick] (cGate.north) to[out=90, in=-110] (fGate.south);
\draw[-{Stealth}, gray!75, semithick] (cVal.east) to[out=90, in=-70] (fVal.south);
\end{tikzpicture}
\caption{The dictionary at a glance: the gated single-head attention layer on the probability simplex, annotated with the circuit realizing each component ($q_i \propto W_Q\sqrt{x_i}$, $k_j \propto W_K\sqrt{x_j}$; at $\eta = \pi/2$, $\beta = 1$ the formula reduces to Eq.~\ref{eq:classical} under the conventions of Remark~\ref{rem:conventions}). Softmax: compiled $\cos^2$ Born readout with $\theta_{ij} = 2\arccos(e^{z_{ij}/2})$ (Proposition~\ref{prop:cossoftmax}); temperature: $L$ post-selected rounds (Corollary~\ref{cor:temperature}); scores: Hadamard test (Lemma~\ref{lem:overlap}); gating and residual: single-ancilla preparation angle (Theorem~\ref{thm:gated-residual}); values: column-loading channel (Theorem~\ref{thm:weighted-sum}).}
\label{fig:headline}
\end{figure*}

\section{Introduction}
\label{sec:intro}

Transformer self-attention~\cite{vaswani2017attention} has become the dominant architecture in modern machine learning, underlying state-of-the-art large language models and a growing fraction of vision, audio, and scientific-computing systems. Its central operation is dot-product attention with a softmax non-linearity over inner-product scores, an additive residual connection~\cite{he2016deep,srivastava2015highway} between sublayer input and output, and a value aggregation step. The success of this architecture has motivated extensive theoretical and empirical studies of its components, including alternatives to the exponential softmax~\cite{choromanski2020rethinking,tsai2019transformer}.

Quantum machine learning~\cite{schuld2018supervised,bromley2020applications} offers an orthogonal computational substrate. Variational quantum circuits operate natively on the unit sphere of a complex Hilbert space, with the Born rule supplying probability normalization without an explicit projection step. This is suggestive: classical transformer attention enforces the simplex structure of its output by applying a softmax, while a quantum circuit's output lives on the simplex by construction. Special cases of full attention bound the total output of the model to a probability simplex or a collection of simplices where all outputs are enforced to sum to one. This probability simplex setting is actively studied classically: flow matching on Riemannian manifolds~\cite{lipman2023flow,chen2024riemannian} extends naturally to the simplex via the Fisher--Rao metric, Dirichlet flow matching~\cite{avdeyev2024dirichlet} parameterizes simplex-valued generative dynamics directly, and continuous diffusion for categorical data~\cite{dieleman2022continuous} operates on simplex-valued representations of discrete tokens. Whether the structural alignment between this regime and quantum amplitude encoding yields a faithful quantum analogue of transformer attention, and at what operational cost, is the question this paper addresses for the single-head case.

\paragraph{Related work on quantum attention.}
Several recent constructions implement transformer attention as a quantum circuit. Cherrat et al.~\cite{cherrat2024qvit} build quantum vision transformers from parametrized orthogonal layers and matrix-loading circuits, with one variant realizing attention via compound matrices. QSAN~\cite{shi2022qsan} replaces classical inner-product similarity with a quantum-logic similarity and represents attention scores as a density matrix produced by circuit evolution. QCSAM~\cite{chen2025qcsam} realizes self-attention via complex linear combinations of unitaries (LCU) with complex-valued attention weights. Guo et al.~\cite{guo2024quantumtransformer} give a fault-tolerant route in which self-attention, residual, and feed-forward layers are computed entirely by block-encoded quantum linear-algebra subroutines on amplitude-encoded inputs. None of these works states the cosine-squared / exponential softmax bijection on the simplex interior, characterizes the additive residual as a $\pi/2$ slice of a single-ancilla LCU family, or assembles these into an exact single-head equivalence theorem; the construction we present is complementary to and distinct from these prior architectures.

\paragraph{Contributions.}
We give a self-contained construction of a single-head gated transformer attention layer in a probability-space setting compatible with quantum amplitude encoding. The construction is built from six primitives: amplitude encoding, block-encoded query/key projections~\cite{gilyen2019quantum,childs2017quantum}, coherent inner-product computation via the Hadamard test, a softmax map realized by Born-rule readout of controlled-$R_y$ rotations (denoted $\mathrm{cR}_y$, where $R_y(\theta) = e^{-i\theta Y/2}$ is the single-qubit $Y$-axis rotation), value aggregation via a controlled column-loading channel (a Stinespring dilation of the column-stochastic $W_V$), and a gated single-ancilla residual realized by a LCU~\cite{childs2012hamiltonian} with parameterized ancilla preparation. Three consequences of this construction are notable:

\begin{enumerate}[leftmargin=*]
\item \textbf{Cosine-squared softmax is isomorphic to exponential softmax on the simplex interior, and strictly extends it at the boundary.} We prove the angle-score bijection $\theta = 2\arccos(e^{z/2})$ realizes the equality of the two softmax maps pointwise on the open simplex. At the boundary, the cosine-squared family expresses exact zeros (sparse attention) at finite parameter values, which the exponential family cannot reach with bounded scores.
\item \textbf{The classical additive residual identity is the $\eta = \pi/2$ special case of a parameterized single-ancilla construction.} A single ancilla prepared in $\cos(\eta/2)|0\rangle + \sin(\eta/2)|1\rangle$, controlled-$U$, Hadamard, and post-select-$|0\rangle$ produces the normalized linear combination $\cos(\eta/2)|X\rangle + \sin(\eta/2) U|X\rangle$, recovering the no-residual ($\eta = 0$), additive ($\eta = \pi/2$), and full-sublayer ($\eta = \pi$) cases as specific gate-angle settings.
\item \textbf{The softmax temperature is measurement repetition.} $L$ identical rounds of the post-selected Born stage realize exp-softmax at inverse temperature $\beta = L$ exactly, and positive measurement weights summing to $\beta$ realize any real temperature (Theorem~\ref{thm:trotter}, Corollary~\ref{cor:temperature}); the temperature axis, a hyperparameter classically, is a physical repetition count in the quantum realization.
\end{enumerate}

The six primitives make up the paper's main result. The composition is stated as Theorem~\ref{thm:equiv}. Sections~\ref{sec:encoding}--\ref{sec:residual} each supply one step of its proof, and Section~\ref{sec:equiv} combines them.

\begin{mainresult}[informal]
Let the inputs, values, and outputs of a single-head attention layer with an additive residual lie on the probability simplex, with queries and keys formed from the amplitude lift $\sqrt{x_i}$ and with scores normalized (Remark~\ref{rem:conventions}). Then, a quantum circuit that uses one measure-and-reload step per attention score and whose learnable parameters are all rotation-gate angles, reproduces the layer's output as its Born-rule readout exactly (in the infinite-shot limit).
\end{mainresult}

\Cref{fig:headline} annotates the layer formula with the circuit realizing each component; Table~\ref{tab:dictionary} in Section~\ref{sec:discussion} collects the full component-by-component dictionary, with the exactness status and certifying result for each entry.

The headline algebraic content (the angle-score bijection, the boundary strict extension, the multi-stage post-selection identity with its inverse-temperature generalization, the gated single-ancilla circuit, the stochastic-value-channel diagonal identity, and the polynomial-exponential impossibility behind the coherent no-go) is machine-checked in Lean 4~\cite{demoura2021lean,mathlib2020} against the mathlib library. The corresponding Lean theorem statements and proofs are in Appendix~\ref{sec:leancode}. A successful Lean kernel build is the type-theoretic certificate that every listed theorem holds.

\paragraph{Scope and limitations.}
Two restrictions are essential to the construction. First, the input-output domain is the probability simplex (or its Born-rule lift to the positive orthant of the unit sphere), not generic feature space; this matches the natural setting for quantum amplitude encoding and resolves several normalization issues that would otherwise complicate the equivalence. Second, the construction uses one \emph{measure-and-reload step} per attention scoring: the Hadamard-test outcome $z_{ij}$ is sampled, a classical function is applied to obtain the controlled-$R_y$ rotation angle $\theta_{ij} = 2\arccos(e^{z_{ij}/2})$, and $\theta_{ij}$ is reloaded as a gate parameter for the softmax step. Measure-and-reload is a standard quantum-classical interface pattern, the same mechanism used in adaptive measurements, magic-state distillation, and quantum error correction. The overhead it introduces (statistical noise from finite shots, classical-control latency, mid-circuit measurement) is in principle separately optimizable, e.g., via amplitude estimation, low-latency control electronics, or parallel state preparations. The equivalence is exact in the infinite-shot limit; a fully-coherent alternative that removes the measure-and-reload step is $\varepsilon$-approximate via QSVT polynomial dressing (Theorem~\ref{thm:qsvt}). Exact finite-depth fully-coherent realization without measure-and-reload is provably impossible for affinely-encoded scores within the QSVT circuit family (Theorem~\ref{thm:no-go}). The construction also relies on standard variational-circuit expressivity assumptions: specifically, that a sufficiently deep hardware-efficient rotation-CRY circuit class realizes a dense subset of $SU(2^{b+1})$, hence of contractions on the data register. The underlying gate set (arbitrary single-qubit rotations together with an entangling two-qubit gate on a connected coupling graph) is universal~\cite{barenco1995elementary,bremner2002practical}; Schuld, Sweke, and Meyer~\cite{schuld2021effect} give the complementary function-space statement that data-reuploading circuits realize truncated Fourier series and are universal approximators on bounded domains. We adopt the density as an architectural assumption on the finite-depth family rather than a free claim; the variational realization of the block-encoded $W$ family (Construction~\ref{con:rot-ansatz}, Lemma~\ref{lem:rot-ansatz-realization}) is conditional on this density. The equivalence theorem itself does not require it: the dilation unitary of Lemma~\ref{lem:embed} admits an exact decomposition into rotation gates and CNOTs by standard synthesis~\cite{shende2006synthesis}, so exactness rests only on the infinite-shot limit of the measure-and-reload step.

\paragraph{Organization.}
Section~\ref{sec:setup} establishes notation, the probability-space assumption, and the quantum register layout. Section~\ref{sec:encoding} defines the Born-rule encoding. Section~\ref{sec:qkv} introduces block-encoded unitary projections. Section~\ref{sec:overlaps} treats coherent inner-product computation. Section~\ref{sec:softmax} establishes the cosine-squared softmax via Born-rule normalization, proves the interior isomorphism with exponential softmax, and characterizes the boundary strict-extension behavior. Section~\ref{sec:weighted-sum} treats value aggregation, and Section~\ref{sec:residual} establishes the gated single-ancilla LCU residual circuit. Section~\ref{sec:equiv} states the single-head master equivalence theorem. Section~\ref{sec:lean} discusses the formal verification. Section~\ref{sec:quantumness} examines the architectural status of the construction. Section~\ref{sec:practicality} discusses practicality and big-O notation scaling behavior for the construction. Section~\ref{sec:discussion} discusses implications, and Section~\ref{sec:conclusion} concludes. A final appendix section with Lean code is included in Appendix~\ref{sec:leancode}.

\section{Setup and assumptions}
\label{sec:setup}

\subsection{Classical single-head attention with residual}

In standard transformer architectures, multi-head attention concatenates several independent single-head computations, each with its own learned $W_Q, W_K, W_V$. The construction in this paper treats a single head; extension to multi-head is the obvious parallel composition (one set of variational parameters and one set of ancilla registers per head). We use the unqualified term ``attention'' throughout to refer to single-head attention.

For input $X \in \mathbb{R}^{n\times d}$ with rows $x_i$, the classical single-head self-attention layer with residual computes
\begin{align}
\label{eq:classical}
Q_i &= W_Q x_i, \quad K_i = W_K x_i, \quad V_i = W_V x_i, \\
A_{ij} &= \softmax_j\!\left(\tfrac{Q_i^\top K_j}{\sqrt{d}}\right), \\
Y_i &= \sum_j A_{ij} V_j, \\
Y_i^{\mathrm{out}} &= \tfrac12(Y_i + x_i).
\end{align}

\subsection{Probability simplex setting}

The classical attention representation shown in Eq.~\ref{eq:classical} does not place any constraint on the rows of $X$. A quantum register, on the other hand, consists of only unit-norm amplitude vectors and a Born-rule readout of the register returns probability vectors. We therefore place the assumption that the classical layer (that our model proposes a quantum equivalence for) lies on the probability simplex including its inputs, values, and outputs. This assumption will be used in the rest of the paper.

\begin{assumption}[Probability-space attention]
\label{assn:simplex}
Inputs $X \in \mathbb{R}^{n \times d}$ are row-normalized probability distributions: each row $x_i \geq 0$ component-wise and $\sum_j x_{ij} = 1$. The amplitude lift $\sqrt{x_i} \in S^{d-1}_+$ lies on the positive orthant of the unit sphere. The value matrix $W_V$ is column-stochastic (each column is a probability vector), so that $V_j = W_V x_j$ remains on the simplex.
\end{assumption}

This assumption matches the natural setting for quantum amplitude encoding of probability data (e.g., images interpreted as probability distributions under the Born rule) and resolves three potential issues simultaneously: row-normalization is automatic; the mixed-state output of partial trace is directly the desired probability vector; the LCU residual operates on the same simplex as the classical residual.

\begin{remark}[Lifted-feature, normalized-score, and row-shift conventions]
\label{rem:conventions}
Three conventions connect Eq.~\ref{eq:classical} to the quantum primitives that follow. (i) \emph{Lifted features.} The block-encoded projections act on the amplitude lift: the circuit computes $W_Q \sqrt{x_i}$ (componentwise square root), not $W_Q x_i$. Throughout, Eq.~\ref{eq:classical} is therefore read with the convention that the feature vector presented to the projections is the lift $\sqrt{x_i}$; equivalently, the classical layer is defined on the amplitude (Hellinger) representation of each simplex point. (ii) \emph{Normalized scores.} Post-selected block encodings output normalized states, so the Hadamard test of Lemma~\ref{lem:overlap} returns the normalized overlap $z_{ij} = \Real\langle q_i|k_j\rangle$ with $\|q_i\| = \|k_j\| = 1$; the score in Eq.~\ref{eq:classical} is read as this normalized bilinear form. The input-dependent norms $\|W_* \sqrt{x_i}\|$ are observable as block-encoding post-selection rates and can be reintroduced classically if unnormalized scores are required, at additional shot cost. The $1/\sqrt d$ temperature is a fixed scalar absorbed into the classical conversion of the measure-and-reload step; equivalently, it can be realized physically as post-selected measurement weight (Corollary~\ref{cor:temperature}). (iii) \emph{Row shift.} The overlap $z_{ij} \in [-1, 1]$ can be positive, while $\theta = 2\arccos(e^{z/2})$ requires $z \leq 0$; the reload step therefore first applies the softmax-invariant row shift $z_{ij} \mapsto z_{ij} - \max_l z_{il}$. Operationally, all $n$ scores of row $i$ are collected before any angle for that row is reloaded.
\end{remark}

\subsection{Quantum registers}

Each of the following sections introduces registers. The total layout of all registers used in the rest of the work is outlined here for reference. Set $a = \lceil \log_2 n\rceil$, $b = \lceil \log_2 d\rceil$. The construction uses:
\begin{center}
\small
\begin{tabular}{@{}l>{\raggedright}p{0.2\columnwidth}>{\raggedright\arraybackslash}p{0.47\columnwidth}@{}}
\toprule
register & qubits & role \\
\midrule
$A$, $A'$ & $a$ each & token query, key indices \\
$B_Q, B_K$ & $b$ each & per-token feature amplitudes (score pathway) \\
$B'$ & $b$ & value-channel feature register, traced out (may recycle $B_K$) \\
$B_V$ & $b$ & value-channel output; diagonal carries $Y_i$ \\
$\mathrm{anc}_Q, \mathrm{anc}_K$ & 1 each, recyclable via reset & block-encoding ancillas (the value channel needs none) \\
$S$ & $a$ & attention-weight register \\
$h$ & 1 & overlap (Hadamard-test) ancilla, reusable \\
$h^{(\ell)}$ & 1, reused & Trotter post-selection ancilla \\
$C$ & 1 & residual gating ancilla (the single ancilla) \\
\bottomrule
\end{tabular}
\end{center}

\section{Encoding (Born-rule lift)}
\label{sec:encoding}

The first primitive we introduce loads the classical input into the quantum register. Under Assumption~\ref{assn:simplex}, each row $x_i$ is a probability vector so its componentwise square root is a unit vector. With that assumption, $\sqrt{x_i}$ can serve directly as the amplitudes of a $b$-qubit state and a noiseless measurement of that state in the computational basis returns $x_i$. The sequence is loaded as a uniform superposition over the token index so that every downstream operation acts on all tokens simultaneously. This state is the input to every later step of Theorem~\ref{thm:equiv}.

\begin{definition}[Sequence-amplitude encoding]
\label{def:encoding}
Define $|X\rangle_{AB} \in \HilbertA \otimes \HilbertB$ by
\begin{equation}
\begin{gathered}
|X\rangle_{AB} \;=\; \frac{1}{\sqrt{n}}\sum_{i=0}^{n-1} |i\rangle_A \otimes |x_i\rangle_B,\\
|x_i\rangle_B \;=\; \sum_{j=0}^{d-1} \sqrt{x_{ij}}\,|j\rangle_B.
\end{gathered}
\end{equation}
The state preparation is unitary; the marginal over $A$ is uniform $1/n$, and the conditional state on $B$ given outcome $i$ is the unit-amplitude lift of token $x_i$.
\end{definition}

The Born-rule simplex constraint $\sum_i |\alpha_i|^2 = 1$ replaces the classical softmax simplex projection automatically: $|x_i\rangle_B$ is unit-norm by construction. This is visualized in \Cref{fig:bornrule} where B represents the Born-rule lift operation state preparation unitary.

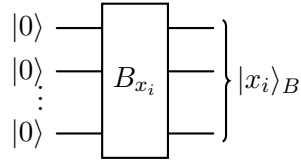
\begin{figure}
\begin{center}
\begin{quantikz}[column sep=0.6cm, row sep=0.25cm]
\lstick{$|0\rangle$} & \gate[4]{B_{x_i}} & \rstick[4]{$|x_i\rangle_B$} \\
\lstick{$|0\rangle$} & & \\
\lstick{$\vdots$}\setwiretype{n} & & \\
\lstick{$|0\rangle$} & &
\end{quantikz}
\end{center}
\caption{Born-rule encoding diagram.}
\label{fig:bornrule}
\end{figure}

\section{Q, K projections via block encoding}
\label{sec:qkv}

With $|x_i\rangle_B$ loaded, the next classical operation is the projection to queries and keys, $W_Q x_i$ and $W_K x_i$. Neither $W_Q$ nor $W_K$ need preserve norm, whereas every unitary does, so a generic real $W \in \{W_Q, W_K\}$ cannot be realized directly as a unitary on $B$. The solution to this is to embed $W$ as a block of a larger unitary acting on $B$ with one ancilla. Post-selecting the ancilla on $|0\rangle$ gives a state which is proportional to $W|x_i\rangle$ on $B$ (block encoding, \Cref{fig:qkvproj}). This section first states the existence of that unitary (Lemma~\ref{lem:embed}) and then gives the trainable circuit family that realizes this (Section~\ref{sec:variational-block-encoding}). These block encodings supply the query and key projections in Theorem~\ref{thm:equiv}.

\begin{lemma}[Block-encoded unitary]
\label{lem:embed}
For any $W \in \mathbb{R}^{d \times d}$ with $\|W\|_{\mathrm{op}} \leq 1$, there exists a unitary $V \in U(2^{b+1})$ on $b+1$ qubits (with $W$ zero-padded to dimension $2^b$ when $d < 2^b$) such that
\begin{equation}
\langle 0|_{\mathrm{anc}} \otimes \langle s'|_B \cdot V \cdot |0\rangle_{\mathrm{anc}} \otimes |s\rangle_B \;=\; W_{s' s},
\end{equation}
i.e., the upper-left $d \times d$ block of $V$ equals $W$.
\end{lemma}
\begin{proof}
$V = \begin{pmatrix} W & * \\ * & * \end{pmatrix}$ exists for any $W$ with $\|W\|_{\mathrm{op}} \leq 1$; the $*$ blocks are constructed by orthogonalizing rows/columns to make $V V^\dagger = I$. The operator-norm bound is enforced classically by absorbing a scale factor: encode $W/\|W\|_{\mathrm{op}}$ and store the scale as a classical multiplier on the readout.
\end{proof}

Applying Lemma~\ref{lem:embed} to $W_Q$ and $W_K$ gives the two block encodings that the score pathway uses.

\begin{corollary}[Q, K via block encodings]
\label{cor:qkv-embed}
For classical weight matrices $W_Q, W_K$, there exist unitaries $\tilde V_Q, \tilde V_K$ (on $b+1$ qubits each, one block-encoding ancilla per matrix) such that, when applied to the encoded input,
\begin{equation}
(\langle 0|_{\mathrm{anc}_*} \otimes I_B) \tilde V_* (|0\rangle_{\mathrm{anc}_*} \otimes |x_i\rangle_B) = W_* |x_i\rangle_B,
\end{equation}
for $* \in \{Q, K\}$, i.e., the post-selected $B$-register state is the linear action of $W_*$ on the encoded amplitude vector $|x_i\rangle_B = \sum_j \sqrt{x_{ij}}|j\rangle$. The value matrix $W_V$ is deliberately not realized by block encoding: the score pathway needs the amplitude-level states $W_*\sqrt{x_i}$, which block encoding provides, whereas value aggregation needs the probability-level $W_V x_j$, which it cannot provide (Construction~\ref{con:value-load}); the column-stochastic $W_V$ is instead realized exactly as a quantum channel in Section~\ref{sec:weighted-sum}.
\end{corollary}

\begin{remark}[Parallel projection]
Applying $\tilde V_Q$ on $B \otimes \mathrm{anc}_Q$ acts in parallel on every $|x_i\rangle$ in the superposition: a classical projection cost of $\Theta(nd^2)$ is replaced by a single unitary application whose gate-count and depth depend on the structure of $W_Q$ via the chosen block-encoding scheme; for a generic dense $W_Q$ the depth scales polynomially in $d$, while sparse or structured $W_Q$ admit $\mathrm{polylog}(d)$-depth implementations.
\end{remark}

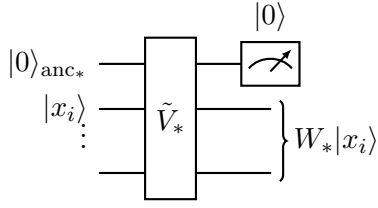
\begin{figure}
\begin{center}
\begin{quantikz}[column sep=0.6cm, row sep=0.25cm]
\lstick{$|0\rangle_{\mathrm{anc}_*}$} & \gate[4]{\tilde V_*} & \meter{|0\rangle} \\
\lstick{$|x_i\rangle$} & & \rstick[3]{$W_*|x_i\rangle$} \\
\lstick{$\vdots$}\setwiretype{n} & & \\
\lstick{} & &
\end{quantikz}
\end{center}
\caption{Q, K projection diagram.}
\label{fig:qkvproj}
\end{figure}

\subsection{Variational realization via rotation-CRY ansatz}
\label{sec:variational-block-encoding}

Lemma~\ref{lem:embed} guarantees a unitary $V$ realizing $W$ as its upper-left block, but is non-constructive. The implementation uses a parameterized circuit family that makes the construction explicit and trainable.

\begin{construction}[Hardware-efficient rotation-CRY ansatz]
\label{con:rot-ansatz}
On $b+1$ qubits (data register $B$ plus one block-encoding ancilla), define $U_L(\boldsymbol\Theta)$ by $L$ repetitions of the layer
\begin{enumerate}[leftmargin=*,nosep]
\item Per-qubit $\mathrm{Rot}(\varphi_{l,q}, \theta_{l,q}, \omega_{l,q}) = R_z(\omega_{l,q})\,R_y(\theta_{l,q})\,R_z(\varphi_{l,q})$ for $q \in \{0, \ldots, b\}$;
\item Anchor entanglers $\mathrm{CRY}(\alpha_{l,q})$ with ancilla controlling each data qubit $q \in \{0, \ldots, b-1\}$;
\item Brick entanglers $\mathrm{CRY}(\beta_{l,q})$ data $q \to q+1$ for $q \in \{0, \ldots, b-2\}$,
\end{enumerate}
followed by a final layer of per-qubit Rot. Define the variational block-encoded operator $W_L(\boldsymbol\Theta) := (\langle 0|_{\mathrm{anc}} \otimes I_B)\, U_L(\boldsymbol\Theta)\, (|0\rangle_{\mathrm{anc}} \otimes I_B)$.
\end{construction}

For Construction~\ref{con:rot-ansatz} to stand in for the unitary of Lemma~\ref{lem:embed}, its post-selected block must always be a valid block encoding. It must also be able to approximate any target $W$. The ansatz starting at the identity also helps with training. The next lemma establishes all three properties.

\begin{lemma}[Variational realization of the block-encoded operator]
\label{lem:rot-ansatz-realization}
The family $\{W_L(\boldsymbol\Theta)\}_{L, \boldsymbol\Theta}$ has the following properties.
\begin{enumerate}[leftmargin=*,nosep]
\item \emph{Validity.} For every $L$ and every $\boldsymbol\Theta$, $\|W_L(\boldsymbol\Theta)\|_{\mathrm{op}} \leq 1$, so $W_L$ is a valid block-encoded operator in the sense of Lemma~\ref{lem:embed}.
\item \emph{Identity at the origin.} $U_L(\mathbf{0}) = I_{2d}$ and $W_L(\mathbf{0}) = I_d$.
\item \emph{Density.} For every contraction $W \in \mathbb{C}^{d \times d}$ with $\|W\|_{\mathrm{op}} \leq 1$ and every $\varepsilon > 0$, there exist $L_\varepsilon \in \mathbb{N}$ and parameters $\boldsymbol\Theta_\varepsilon$ such that $\|W_{L_\varepsilon}(\boldsymbol\Theta_\varepsilon) - W\|_{\mathrm{op}} < \varepsilon$.
\end{enumerate}
\end{lemma}
\begin{proof}
Part (1) is immediate: the top-left block of any unitary in $U(2d)$ acts as a contraction on $\mathbb{C}^d$ — for unit $X \in \mathbb{C}^d$, lifting to $Y = (X, 0) \in \mathbb{C}^{2d}$ gives $\|Y\| = 1$ and $\|U_L Y\| = 1$ (unitarity), so by Pythagoras $\|W_L X\|^2 = \|P_0 U_L Y\|^2 \leq \|U_L Y\|^2 = 1$.

Part (2) is algebraic: $R_y(0) = R_z(0) = I_2$ (from $\cos 0 = 1$, $\sin 0 = 0$, $e^0 = 1$) and $\mathrm{CRY}(0) = I_4$ identically, so each gate at zero parameter is identity and their product on $b+1$ qubits is $I_{2d}$. Hence $W_L(\mathbf{0}) = (\langle 0|_{\mathrm{anc}} \otimes I_B) I_{2d} (|0\rangle_{\mathrm{anc}} \otimes I_B) = I_d$. The algebraic core of part (2) is Lean-verified in Appendix~\ref{sec:leancode}.

Part (3) is the density of the variational family in the closed unit ball of $d \times d$ contractions. The gate set generating each layer (arbitrary single-qubit rotations plus CRY entanglers on a connected coupling graph) is universal~\cite{barenco1995elementary,bremner2002practical}, so the rotation-CRY family $\{U_L(\boldsymbol\Theta)\}$ is dense in $SU(2^{b+1})$ in the operator-norm topology as $L \to \infty$. By unitary-dilation (Sz.-Nagy--Halmos) surjectivity, the top-left $d \times d$ block of $U(2d)$ traces out exactly the closed unit ball of contractions in $\mathbb{C}^{d \times d}$, hence $\{W_L\}$ is dense in this ball. Density is a topological / Lie-group statement and is not formalized in Lean here.
\end{proof}

The ancilla's per-qubit $\mathrm{Rot}$ gates are essential: without them, the ancilla would remain in $|0\rangle$ throughout, post-selection would succeed with probability 1, and $W_L \equiv I$ for every parameter setting. The anchor CRY gates combined with ancilla rotations are precisely what gives the variational family access to non-trivial contractions.

\section{Pairwise overlaps via Hadamard test}
\label{sec:overlaps}

The projections defined in Section~\ref{sec:qkv} place the query and key states $|q_i\rangle$ and $|k_j\rangle$ on the data register. The next classical operation takes their inner product. An inner product can be realized through interference by preparing the two states on two branches of a single ancilla qubit and recombining the branches. The ancilla's outcome probability is then the affine function of $\Real\langle q_i|k_j\rangle$. This is the Hadamard test~\cite{Buhrman2001} (\Cref{fig:swaptest}). The overlap it estimates is the score that the measure-and-reload step of Theorem~\ref{thm:equiv} converts into a rotation angle.

\begin{lemma}[State-overlap Hadamard test]
\label{lem:overlap}
Let $B$ be a single data register, $h$ a single-qubit ancilla, and let $\mathrm{ctrl\text{-}prep}_{ij}$ denote the unitary that, conditioned on token indices $|i\rangle_{A}|j\rangle_{A'}$, loads $|q_i\rangle_B$ when the ancilla is $|0\rangle_h$ branch and $|k_j\rangle_B$ on the $|1\rangle_h$ branch. The circuit $H_h \cdot \mathrm{ctrl\text{-}prep}_{ij} \cdot H_h$ applied to $|i,j\rangle|0\rangle_B|0\rangle_h$ produces
\begin{equation}
|i,j\rangle \otimes \tfrac12\big(|0\rangle_h(|q_i\rangle + |k_j\rangle)_B + |1\rangle_h(|q_i\rangle - |k_j\rangle)_B\big),
\end{equation}
with marginal probability $|0\rangle_h$
\begin{equation}
P_{ij}(|0\rangle_h) \;=\; \tfrac14\big\|\,|q_i\rangle + |k_j\rangle\,\big\|^2 \;=\; \tfrac12\big(1 + \Real\langle q_i | k_j\rangle\big).
\end{equation}
\end{lemma}
\begin{proof}
Starting from $|0\rangle_h|0\rangle_B$, the first $H_h$ produces
\begin{equation*}
\tfrac{1}{\sqrt 2}(|0\rangle_h + |1\rangle_h)\otimes|0\rangle_B.
\end{equation*}
Applying $\mathrm{ctrl\text{-}prep}_{ij}$ branch-by-branch yields $\tfrac{1}{\sqrt 2}(|0\rangle_h|q_i\rangle_B + |1\rangle_h|k_j\rangle_B)$, and the second $H_h$ gives the displayed branch decomposition. The marginal $|0\rangle_h$ probability is the squared norm of the $|0\rangle_h$-branch data coefficient:
\begin{align*}
\tfrac14\|\,|q_i\rangle + |k_j\rangle\,\|^2 &= \tfrac14(2 + 2\Real\langle q_i|k_j\rangle) \\
&= \tfrac14\bigl(\|q_i\|^2+\|k_j\|^2 + 2\Real\langle q_i|k_j\rangle\bigr) \\
&= \tfrac12(1+\Real\langle q_i|k_j\rangle),
\end{align*}
where the last equality uses $\|q_i\|^2=\|k_j\|^2=1$.
\end{proof}

\begin{remark}[Oracle status of $\mathrm{ctrl\text{-}prep}$ and the composed variant]
\label{rem:ctrl-prep-oracle}
Lemma~\ref{lem:overlap} treats $\mathrm{ctrl\text{-}prep}_{ij}$ as an oracle for the \emph{normalized} states $|q_i\rangle, |k_j\rangle$. In the composed circuit these are post-selected block-encoding outputs, and post-selection cannot be nested inside a controlled unitary. The operational circuit instead applies controlled-$\tilde V_Q$ and controlled-$\tilde V_K$ (which are unitaries) and post-selects the block-encoding ancilla once at the end; conditional on that success, the ancilla marginal becomes
\begin{equation}
\begin{gathered}
P_{ij}(|0\rangle_h) \;=\; \frac12\left(1 + \frac{2\, a_q a_k\, \Real\langle q_i|k_j\rangle}{a_q^2 + a_k^2}\right),\\
a_q = \|W_Q\sqrt{x_i}\|,\quad a_k = \|W_K\sqrt{x_j}\|,
\end{gathered}
\end{equation}
with $|q_i\rangle, |k_j\rangle$ normalized as in the lemma. The norms $a_q^2, a_k^2$ are separately observable as single-branch post-selection rates (Remark~\ref{rem:conventions}(ii)), so the normalized score $z_{ij} = \Real\langle q_i|k_j\rangle$ is classically recoverable within the same measure-and-reload step, at additional shot cost. The lemma as stated is the idealized oracle form; the composed form differs from it only by this classically invertible reweighting.
\end{remark}

\begin{figure}
\begin{center}
\resizebox{\columnwidth}{!}{%
\begin{quantikz}[column sep=0.3cm, row sep=0.25cm]
\lstick{$|0\rangle_h$} & \gate{H} & \ctrl{1} & \gate{H} & \meter{} \\
\lstick{$|0\rangle_B$} & & \gate[3]{\mathrm{ctrl\text{-}prep}_{ij}} & & \rstick[3]{\shortstack{$\propto |q_i\rangle{+}|k_j\rangle$\\on $h{=}0$ branch}} \\
\lstick{$\vdots$}\setwiretype{n} & & & & \\
\lstick{} & & & &
\end{quantikz}}
\end{center}
\caption{Hadamard-/swap-test diagram.}
\label{fig:swaptest}
\end{figure}
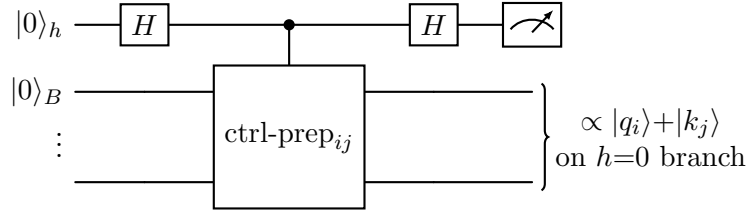

A SWAP-test variant on separate registers $B_Q, B_K$ (using $H_h \cdot \cSWAP_{h \to (B_Q, B_K)} \cdot H_h$) yields instead $P(|0\rangle_h) = \tfrac12(1 + |\langle q_i|k_j\rangle|^2)$, the squared overlap. For real-amplitude states such as those arising from amplitude encoding of probability vectors, $\langle q_i|k_j\rangle$ is real, so squared-overlap and signed-overlap differ by a sign that the SWAP test discards; the controlled-prep variant in Lemma~\ref{lem:overlap} is therefore the appropriate route when the score $z_{ij} = \langle q_i|k_j\rangle$ itself (and not $z_{ij}^2$) is the desired output.

The algebraic identity $|(1+z)/\sqrt{2}|^2 = (1 + 2\Real(z) + |z|^2)/2$ is a partial step in this calculation: it gives the squared norm of the scalar combination $(1+z)/\sqrt 2$. The full Hadamard-test probability $\tfrac12(1+\Real(z))$ follows from this identity together with unitarity, which contributes a $-|z|^2$ term from the $|1\rangle_h$-branch component that cancels the $|z|^2$ in the algebraic identity.

\section{Born-rule normalization is a softmax}
\label{sec:softmax}

The previous step gives $z_{ij}$ (the result of the inner product of $q_i$ and $k_j$). The next classical step is normalization of each row by exp-softmax. There is a quantum equivalent to this row-normalization which is achieved by preparing a state, conditioned on the token index $|i\rangle_A$, whose amplitudes over $j$ carry the scores (\Cref{fig:qattn}). This section shows that the resulting family is exp-softmax under a change of variables and that it reaches distributions that exp-softmax cannot at finite scores. We begin this section with the attention map and its Born readout. Proposition~\ref{prop:cossoftmax} and Theorem~\ref{thm:isomorphism} are the softmax step of Theorem~\ref{thm:equiv}. The rest of the section outlines the consequences of the normalization primitive (temperature as a measurement count) and where its limits lie (the coherent no-go and its QSVT approximation).

\begin{definition}[Quantum attention map]
\label{def:qattn}
Let $\theta_{ij}$ be an angle parameter, supplied either by the measure-and-reload step (Theorem~\ref{thm:equiv}) or as a free learnable parameter (Theorem~\ref{thm:strict-extension}). Define the unitary $U_{\mathrm{attn}}$ on $\HilbertA \otimes \HilbertS$ by
\begin{equation}
U_{\mathrm{attn}} |i\rangle_A |0\rangle_S \;=\;
|i\rangle_A \sum_j \widetilde a_{ij} |j\rangle_S,
\quad \sum_j |\widetilde a_{ij}|^2 = 1.
\end{equation}
\end{definition}

Definition~\ref{def:qattn} leaves the amplitudes $\widetilde a_{ij}$ free. Taking them proportional to the $|0\rangle$-amplitude of $R_y(\theta_{ij})|0\rangle$ (which is $\cos(\theta_{ij}/2)$) lets the angles of controlled-$R_y$ gates set the attention pattern directly.

\begin{proposition}[$\cos^2$-softmax]
\label{prop:cossoftmax}
For controlled-$R_y$ rotation angles $\theta_{ij}$ on $S$ indexed by $A$, with $\widetilde a_{ij} \propto \cos(\theta_{ij}/2)$, the Born-rule attention pattern is
\begin{equation}
\boxed{\quad P(j \mid i) \;=\; \frac{\cos^2(\theta_{ij}/2)}{\sum_\ell \cos^2(\theta_{i\ell}/2)} \quad}
\end{equation}
which is row-stochastic by construction (Born-rule simplex).
\end{proposition}
\begin{proof}
The angles $\theta_{ij}$ are classical at compile time (reloaded measure-and-reload outputs, or free parameters), so the normalized amplitudes $\widetilde a_{ij} = \cos(\theta_{ij}/2)\big/\sqrt{\sum_\ell \cos^2(\theta_{i\ell}/2)}$ are classically computable, and $U_{\mathrm{attn}}$ is compiled as an $i$-controlled state preparation realizing them on the $a$-qubit register $S$ (a standard rotation-tree synthesis with $O(n)$ gates per row); no post-selection is involved. Born-rule readout of $S$ then returns $j$ with probability $|\widetilde a_{ij}|^2 = \cos^2(\theta_{ij}/2)/\sum_\ell \cos^2(\theta_{i\ell}/2)$, the boxed pattern. An alternative realization applies a literal controlled-$R_y(\theta_{ij})$ per branch on a fresh ancilla and post-selects; see Remark~\ref{rem:trotter-special-cases}.
\end{proof}

We represent this by a single unitary shown in \Cref{fig:qattn}.

\begin{figure}
\begin{center}
\begin{quantikz}[column sep=0.6cm]
\lstick{$|i\rangle_A$} & \ctrl{1} & \\
\lstick{$|0\rangle_S^{\otimes a}$} & \gate{U_{\mathrm{attn}}(\boldsymbol\theta_i)} & \meter{} \rstick{$j \sim P(j\mid i)$}
\end{quantikz}
\end{center}
\caption{Quantum attention diagram.}
\label{fig:qattn}
\end{figure}
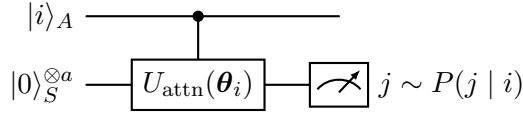

\subsection{$\cos^2$- and exp-softmax: isomorphism on the interior, strict extension on the closure}

The two softmax forms agree on the interior of the simplex but differ at the boundary. We make both statements precise.

\begin{theorem}[Interior isomorphism]
\label{thm:isomorphism}
On the interior of the row-stochastic simplex (where every entry satisfies $A_{ij} > 0$ strictly), the $\cos^2$-softmax and exp-softmax distribution families are equal as sets, related by the angle-score bijection
\begin{equation}
\theta_{ij} = 2\arccos\!\left(e^{z_{ij}/2}\right), \qquad
z_{ij} = 2\log\!\cos(\theta_{ij}/2),
\end{equation}
defined for $z_{ij} \in (-\infty, 0]$ and $\theta_{ij} \in [0, \pi)$. Under this bijection, the single-stage $\cos^2$-softmax exactly equals the exp-softmax pointwise:
\begin{equation}
\frac{\cos^2(\theta_{ij}/2)}{\sum_l \cos^2(\theta_{il}/2)}
\;=\; \frac{e^{z_{ij}}}{\sum_l e^{z_{il}}}.
\end{equation}
\end{theorem}
\begin{proof}
Substituting $\cos^2(\theta_{ij}/2) = e^{z_{ij}}$ into the $\cos^2$-softmax form gives the exp-softmax form directly. The map is a bijection on the open domain $z_{ij} \in (-\infty, 0]$ to $\theta_{ij} \in [0, \pi)$, with softmax shift invariance ensuring $z_{ij} \leq 0$ is achievable without loss of expressivity.
\end{proof}

Theorem~\ref{thm:isomorphism} pairs every finite score $z_{ij} \leq 0$ with an angle in $[0, \pi)$. The endpoint $\theta_{ij} = \pi$ corresponds to no finite score, and it is the point where the two families part.

\begin{theorem}[$\cos^2$ strictly extends exp at the boundary]
\label{thm:strict-extension}
With free angle parameterization $\theta_{ij} \in [0, \pi]$, restricted to rows with at least one $\theta_{il} \neq \pi$ so that every row distribution is defined, the $\cos^2$-softmax distribution family reaches the closure of the simplex, including boundary distributions where some entries $A_{ij}$ are exactly $0$. With finite-valued score parameterization $z_{ij} \in \mathbb{R}$, the exp-softmax distribution family reaches only the interior of the simplex ($A_{ij} > 0$ for every $i, j$). Therefore
\begin{equation}
\begin{gathered}
\{P_{\cos^2}(\theta) : \theta \in [0,\pi]^{n\times n}\}
\\ {}\supsetneq\;
\{P_{\exp}(z) : z \in \mathbb{R}^{n\times n}\},
\end{gathered}
\end{equation}
the inclusion being strict.
\end{theorem}
\begin{proof}
Setting $\theta_{ij} = \pi$ for any specific $(i, j)$ gives $\cos^2(\pi/2) = 0$, so $A_{ij} = 0$ at finite parameter value (the Lean witness sets $\theta_{ij} = \pi$ at a single entry and $0$ elsewhere in the row, so the row denominator is $n - 1 > 0$). By contrast, $e^{z_{ij}} > 0$ for all finite $z_{ij}$, so $A_{ij} > 0$ strictly under exp-softmax with finite parameters. Boundary distributions (where some entries are exactly zero) are witness to the strict inclusion.
\end{proof}

A single-stage $\cos^2$-softmax with free angle parameterization has strictly greater expressivity than single-stage exp-softmax with finite scores: it can express ``hard'' attention patterns where some weights are exactly zero, which exp-softmax can only approximate asymptotically. On the interior, the two are isomorphic representations. The $\cos^2$ primitive is therefore a natural extension to the quantum setting that subsumes exp-softmax as a special case (interior distributions) while adding boundary expressivity (sparse / hard-attention distributions).

\subsection{Trotter as one of infinitely many factorizations}

The single-stage equivalence above is the simplest realization. Multi-stage post-selected constructions provide a continuous family of equivalent realizations, with the equal-angle Trotter form (\Cref{fig:trotter}) being the canonical and simplest member.

\begin{construction}[$L$-stage post-selected factorization]
\label{con:trotter}
For score $z_{ij} \leq 0$, integer $L \geq 1$, and any positive weights $\{\alpha_\ell\}_{\ell=1}^L$ with $\sum_\ell \alpha_\ell = \beta$ for a target inverse temperature $\beta > 0$, set the per-stage angle
\begin{equation}
\theta_{ij}^{(\ell)} = 2\arccos\!\left(e^{\alpha_\ell z_{ij}/2}\right), \quad \ell = 1, \dots, L.
\end{equation}
At each stage, prepare a fresh ancilla $h^{(\ell)}$ in $|0\rangle$, apply the controlled rotation $\mathrm{cR}_y(\theta_{ij}^{(\ell)})$ on $h^{(\ell)}$ conditioned on $|i,j\rangle$, measure, post-select $|0\rangle$.
\end{construction}

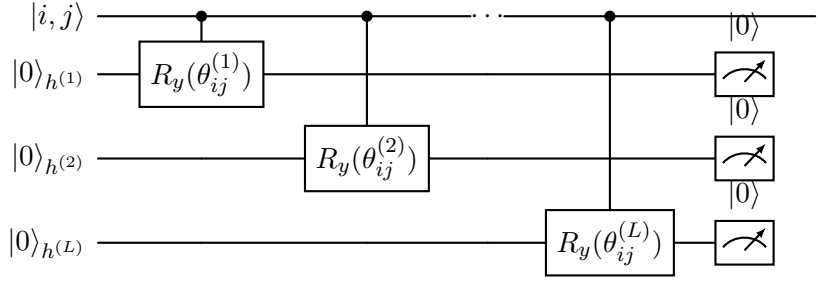
\begin{figure}
\begin{center}
\resizebox{\columnwidth}{!}{%
\begin{quantikz}[column sep=0.3cm, row sep=0.25cm]
\lstick{$|i,j\rangle$} & \ctrl{1} & \ctrl{2} & \cdots & \ctrl{3} & & \\
\lstick{$|0\rangle_{h^{(1)}}$} & \gate{R_y(\theta^{(1)}_{ij})} & & & & \meter{|0\rangle} \\
\lstick{$|0\rangle_{h^{(2)}}$} & & \gate{R_y(\theta^{(2)}_{ij})} & & & \meter{|0\rangle} \\
\lstick{$|0\rangle_{h^{(L)}}$} & & & & \gate{R_y(\theta^{(L)}_{ij})} & \meter{|0\rangle}
\end{quantikz}}
\end{center}
\caption{Trotter factorization diagram.}
\label{fig:trotter}
\end{figure}

Construction~\ref{con:trotter} splits one post-selection into $L$ stages. The next theorem shows that the stages form the exp-softmax at inverse temperature $\beta = \sum_\ell \alpha_\ell$, for any number of stages and any positive split of the weights.

\begin{theorem}[Exact exp-softmax at inverse temperature $\beta$ via $L$-stage post-selection]
\label{thm:trotter}
Under Construction~\ref{con:trotter} with total weight $\sum_\ell \alpha_\ell = \beta > 0$, conditional on all $L$ post-selections succeeding,
\begin{equation}
P(j \mid i) \;=\; \frac{e^{\beta z_{ij}}}{\sum_l e^{\beta z_{il}}} \;=\; \softmax_j(\beta\, z_{ij}).
\end{equation}
The equivalence is exact at any finite $L \geq 1$ and for any positive weights; the case $\beta = 1$ recovers the unit-temperature exp-softmax of Eq.~\ref{eq:classical}.
\end{theorem}
\begin{proof}
Stage $\ell$ post-selection probability per pair: $\cos^2(\theta_{ij}^{(\ell)}/2) = e^{\alpha_\ell z_{ij}}$. Independent post-selections multiply: $\prod_\ell e^{\alpha_\ell z_{ij}} = e^{(\sum_\ell \alpha_\ell)z_{ij}} = e^{\beta z_{ij}}$. Conditioning on $i$ and normalizing across $j$ gives the exp-softmax at inverse temperature $\beta$.
\end{proof}

Setting every weight $\alpha_\ell = 1$ makes the inverse temperature equal to the number of measurement rounds.

\begin{corollary}[Measurement count as inverse temperature]
\label{cor:temperature}
Take $L$ identical rounds of the unit-weight stage ($\alpha_\ell = 1$, per-stage angles $\theta_{ij} = 2\arccos(e^{z_{ij}/2})$). Conditional on all $L$ post-selections succeeding,
\begin{equation}
\begin{gathered}
P_L(j \mid i) \;=\; \frac{e^{L z_{ij}}}{\sum_l e^{L z_{il}}} \;=\; \softmax_j(\beta\, z_{ij}),\\
\beta = L.
\end{gathered}
\end{equation}
The number of repeated Born measurements is the inverse temperature of the attention distribution, exactly at every finite $L$; non-integer $\beta$ is reached by any positive weights summing to $\beta$ (Theorem~\ref{thm:trotter}). The endpoints anchor exactly: zero rounds leave the uniform superposition ($\beta = 0$, infinite temperature), and $\beta \to \infty$ is the hard-attention limit.
\end{corollary}

Three consequences follow. First, the $1/\sqrt{d}$ temperature convention of Remark~\ref{rem:conventions}(ii) acquires a native physical realization as post-selected measurement weight, rather than a classical rescaling of the score. Second, the per-row post-selection success probability is $Z_i(\beta)/n$ with $Z_i(\beta) = \sum_l e^{\beta z_{il}}$: the post-selection rate measures the Gibbs partition function at inverse temperature $\beta$, so cooling the attention distribution (lowering its entropy) is paid for in rejected shots, at per-pair success rate $e^{\beta z_{ij}}$. Third, exp-softmax reaches the simplex boundary (hard attention) only as $\beta \to \infty$, that is, only after infinitely many measurement rounds, whereas the free-angle $\cos^2$ family reaches the boundary at the single finite parameter $\theta = \pi$ (Theorem~\ref{thm:strict-extension}); the boundary strict extension is thus the statement that the quantum parameterization reaches the zero-temperature endpoint at finite parameter values. Physically, each round multiplies the branch amplitudes by $e^{z_{ij}/2}$, one step of imaginary-time evolution under the effective Hamiltonian $H = -z$. The correspondence is specific to the post-selected route: in the measure-and-reload route, $\beta$ is simply a classical scalar multiplied into the score before the $\arccos$.

\begin{remark}[Special cases]
\label{rem:trotter-special-cases}
The construction includes $L=1$, $\alpha_1 = 1$ as the single-stage post-selected case: one round of post-selection on a single ancilla, realizing the same distribution as the compiled state preparation of Proposition~\ref{prop:cossoftmax} via $\theta = 2\arccos(e^{z/2})$, at per-row success probability $Z_i/n$ with $Z_i = \sum_l e^{z_{il}} \in [1, n]$ after the row shift of Remark~\ref{rem:conventions}. Run over the query-index superposition in $A$, this post-selection tilts the $A$-marginal from uniform to $\propto Z_i$ while leaving each conditional $P(j \mid i)$ exact. The construction further includes $L \geq 2$, $\alpha_\ell = 1/L$ as the standard Trotter form (equal angles, simplest multi-stage construction), and any other positive partition as a non-uniform variant. The equivalence is mathematically identical across all choices; operational considerations (depth, success probability, robustness) determine which is preferred.
\end{remark}

\subsection{Why multi-stage post-selection might still be useful}

Although mathematically not unique, multi-stage post-selection has operational features that distinguish it:
\begin{enumerate}[leftmargin=*]
\item \textbf{Smaller per-stage rotation angles} ($\theta = 2\arccos(e^{\alpha z/2})$ with small $\alpha$) keep individual rotations near $\pi/2$, benign for noisy hardware and trainability.
\item \textbf{Multiplicative probability composition} matches the $\big(\prod e^{\alpha_\ell z}\big) = e^z$ structure directly, useful for explicitly Trotter-style architectures.
\item \textbf{Errors partially average out} across stages in the presence of noise.
\end{enumerate}
None of these are mathematical advantages; they are implementation-time properties. With ideal classical simulation or noiseless hardware, the single-stage construction of Theorem~\ref{thm:isomorphism} is operationally simpler.

\subsection{Coherent finite-depth no-go (qualified)}

Theorem~\ref{thm:equiv} uses one measure-and-reload step per attention score. This subsection asks whether that step can be removed. For scores that enter through an affine encoding, no finite-depth QSVT circuit reproduces exp-softmax exactly (Theorem~\ref{thm:no-go}), but QSVT reaches it to any accuracy $\epsilon$ (Theorem~\ref{thm:qsvt}).

\begin{theorem}[Coherent no-go for affinely-encoded scores]
\label{thm:no-go}
Fix $c > 0$ and suppose the scores enter coherently through an \emph{affine} encoding: the block-encoded operator presents singular values $\sigma_{ij} = f(z_{ij})$ with $f$ affine and non-constant (e.g., $\sigma = (z+c)/(2c)$ on $z \in [-c, 0]$), and the score domain contains an open interval. Let $U_L$ be a QSVT circuit of depth $L$ on this block encoding, so that its output amplitudes are $P(\sigma_{ij})$ for a (definite-parity) polynomial $P$ of degree $\leq L$~\cite{gilyen2019quantum}, followed by Born-rule readout $P(\sigma_{ij})^2/\sum_l P(\sigma_{il})^2$. No finite $L$ realizes the exp-softmax $z \mapsto e^{z_{ij}}/\sum_l e^{z_{il}}$ exactly on all score matrices with entries in the domain.
\end{theorem}
\begin{proof}
Exact realization on the row with all entries equal forces $P(f(z)) \neq 0$ throughout the domain (a zero would produce attention weight $0 \neq 1/n$). Now hold every entry of a row except one at a fixed reference score $z_0$ and vary the remaining entry: exactness of the readout forces the ratio $P(f(z))^2/P(f(z_0))^2 = e^{z}/e^{z_0}$, i.e., $P(f(z))^2 = C\,e^{z}$ with $C = P(f(z_0))^2\, e^{-z_0} > 0$ a constant, for every $z$ in an open interval. Writing $\sigma = f(z)$ with $f$ affine and invertible, the polynomial $Q(\sigma) := P(\sigma)^2$ of degree $2L$ agrees with $C e^{a\sigma + b}$, $a \neq 0$, on an open interval of $\sigma$. Agreement of smooth functions on an open interval forces equality of all derivatives at any interior point $\sigma_0$; but $Q^{(2L+1)} \equiv 0$ while $\frac{d^{2L+1}}{d\sigma^{2L+1}}\, C e^{a\sigma+b}\big|_{\sigma_0} = C\, a^{2L+1} e^{a\sigma_0+b} \neq 0$, a contradiction.
\end{proof}

The algebraic core (a polynomial cannot agree with a genuine exponential on an open interval) is Lean-verified in Appendix~\ref{sec:leancode}. The formal proof uses an even shorter route than the derivative count above: differentiating the assumed identity once gives $Q' = aQ$ on the interval, so the polynomial $Q' - aQ$ has infinitely many roots and is identically zero, contradicting $\deg(Q') < \deg(Q)$ for the nonvanishing $Q$.

Theorem~\ref{thm:no-go} rules out exactness, not approximation. Applying a polynomial approximation of the exponential through QSVT gives a fully coherent route with error $\epsilon$.

\begin{theorem}[$\epsilon$-approximation via QSVT for quantum-loaded scores]
\label{thm:qsvt}
Suppose the scores are bounded, $|z_{ij}| \leq c$, and are not extracted via the measure-and-reload step but must be processed coherently in-circuit (no mid-circuit measurement). Then there exists a quantum circuit of depth $O\!\big(\sqrt{\max\{c, \log(1/\epsilon)\}\,\log(1/\epsilon)}\big)$ above the $\cos^2$ baseline that approximates $\exp$-softmax to accuracy $\epsilon$ via Chebyshev polynomial dressing of $\cos(\theta/2)$ (QSVT with one extra ancilla). This is the fully-coherent alternative to the measure-and-reload realization in Theorem~\ref{thm:equiv}; the trade-off is exactness (measure-and-reload) versus full coherence ($\varepsilon$-approximate QSVT).
\end{theorem}
\begin{proof}[Proof sketch]
On $[-1, 1]$ the rescaled exponential $x \mapsto e^{c(x-1)}$ equals $e^{-y}$ under $y = c(1-x) \in [0, 2c]$, and $e^{-y}$ on $[0, 2c]$ admits a polynomial approximation of degree $O\!\big(\sqrt{\max\{c, \log(1/\epsilon)\}\,\log(1/\epsilon)}\big)$ with uniform error $\epsilon$~\cite[Theorem~4.1]{sachdeva2014faster}; the $\sqrt{c}$ dependence is optimal among polynomials by the Markov-type lower bounds surveyed there, and such polynomials are part of the standard QSVT toolbox~\cite{gilyen2019quantum}. The polynomial must also satisfy the usual QSVT boundedness and parity constraints. QSVT applies this polynomial to the singular value $\cos(\theta/2)$ of the block-encoded score rotation using one extra ancilla, producing amplitudes $\epsilon$-close (uniformly) to $e^{z/2}$ up to a global normalization; Born-rule readout of the dressed register then yields each attention row to accuracy $O(\epsilon)$ in total variation.
\end{proof}

\begin{remark}[The encoding is where the hardness lives]
\label{rem:encoding-dependence}
Theorem~\ref{thm:no-go} cannot be strengthened to arbitrary encodings: under the reparameterized encoding $\sigma = e^{z/2}$, which is precisely the angle-score bijection of Theorem~\ref{thm:isomorphism}, the degree-one polynomial $P(\sigma) = \sigma$ already satisfies $P(\sigma)^2 = e^{z}$, and Born readout returns exp-softmax exactly. The obstruction sits entirely in the non-polynomial reparameterization $z \mapsto e^{z/2}$, and the two realizations in this paper are the two ways of paying for it: classically, as one $\arccos$ per score in the measure-and-reload step (exact; Theorem~\ref{thm:equiv}), or coherently, as QSVT polynomial dressing from an affine encoding ($\varepsilon$-approximate at degree $O(\sqrt{\max\{c, \log(1/\varepsilon)\}\,\log(1/\varepsilon)})$; Theorem~\ref{thm:qsvt}). Theorems~\ref{thm:no-go} and~\ref{thm:qsvt} are in this sense a matched pair: the same map that no finite-degree polynomial reaches exactly is reached to accuracy $\varepsilon$ at the cited degree.
\end{remark}

\section{Weighted-sum value loading}
\label{sec:weighted-sum}

The attention register now holds amplitudes whose Born-rule probabilities are the attention weights $A_{ij}$. The next classical layer step is the weighted sum $Y_i = \sum_j A_{ij} W_V x_j$. This requires a value projection, similar to the query and key projections. This is supplied by controlled state preparation and does not require any post-selection. The difference between the value projection and query and key projections of Section~\ref{sec:qkv} is that here the value projection must act on the probabilities $x_j$ rather than on the amplitudes $\sqrt{x_j}$. That is the reason that $W_V$ is realized as a quantum channel rather than as a block encoding (\Cref{fig:valueload}). Theorem~\ref{thm:weighted-sum} below is the value-aggregation step of Theorem~\ref{thm:equiv}.

\begin{construction}[Stochastic value channel]
\label{con:value-load}
After the attention map produces $|i\rangle_A \sum_j \widetilde a_{ij} |j\rangle_S |0\rangle_{B'} |0\rangle_{B_V}$, apply two unitaries. First, the $j$-controlled token encoding $U_{\mathrm{enc}}: |j\rangle_S |0\rangle_{B'} \to |j\rangle_S |x_j\rangle_{B'}$ (the encoding oracle of Definition~\ref{def:encoding}, conditioned on $S$). Second, the value channel
\begin{equation}
\begin{gathered}
U_V:\ |m\rangle_{B'}\,|0\rangle_{B_V} \;\mapsto\; |m\rangle_{B'} \otimes |w_m\rangle_{B_V},\\
|w_m\rangle_{B_V} = \sum_k \sqrt{(W_V)_{km}}\,|k\rangle,
\end{gathered}
\end{equation}
which conditions on the \emph{feature} index $m$ and loads the amplitude lift of the $m$-th column of $W_V$; the column lifts $|w_m\rangle$ are valid unit states because $W_V$ is column-stochastic (Assumption~\ref{assn:simplex}). $U_V$ is a Stinespring dilation of the classical channel $W_V$: tracing out $B'$ leaves on $B_V$, for each token $j$, a mixed state whose computational-basis diagonal is exactly $(W_V x_j)_k$ (Theorem~\ref{thm:weighted-sum}), at unit success probability with no post-selection. This is shown in \Cref{fig:valueload}. The rotation angles of the $d$ column state-preparation trees are the learnable parameters $\boldsymbol\Theta_V$; column-stochasticity of the realized $W_V$ holds automatically for every angle setting (the columns are Born-squared amplitudes), and $U_V$ is input-independent, compiled once per trained model. A block encoding of $W_V$ would not serve here: it acts at the amplitude level, producing $W_V\sqrt{x_j}$, whereas value aggregation requires the probability-level $W_V x_j$, and $W_V\sqrt{x} \neq \sqrt{W_V x}$ in general.
\end{construction}

\begin{figure}
\begin{center}
\begin{quantikz}[column sep=0.6cm, row sep=0.25cm]
\lstick{$|j\rangle_S$} & \ctrl{1} & & \trash{} \\
\lstick{$|0\rangle_{B'}^{\otimes b}$} & \gate{U_{\mathrm{enc}}} & \ctrl{1} & \trash{} \\
\lstick{$|0\rangle_{B_V}^{\otimes b}$} & & \gate{U_V} & \rstick{$\rho_i$, $\mathrm{diag}\,\rho_i = Y_i$}
\end{quantikz}
\end{center}
\caption{Stochastic value channel diagram.}
\label{fig:valueload}
\end{figure}
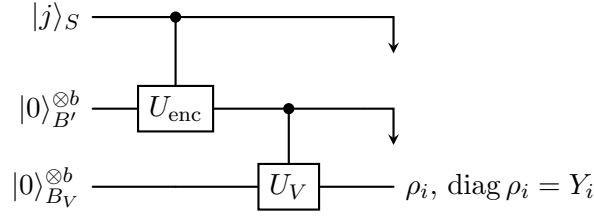

After the partial trace, the attention weights and the columns of $W_V$ combine on the diagonal of $B_V$ exactly as in $Y_i = \sum_j A_{ij} V_j$.

\begin{theorem}[Weighted-sum equivalence]
\label{thm:weighted-sum}
Tracing out $S$ and $B'$, the reduced state on $B_V$ conditional on outcome $i$ in $A$ is
\begin{equation}
\rho_i \;=\; \sum_j |\widetilde a_{ij}|^2 \sum_m x_{jm}\, |w_m\rangle\!\langle w_m|.
\end{equation}
Under Assumption~\ref{assn:simplex}, the diagonal of $\rho_i$ in the computational basis equals the classical attention output $Y_i = \sum_j A_{ij} V_j$, $V_j = W_V x_j$, as a probability vector.
\end{theorem}
\begin{proof}
After the two controlled loads of Construction~\ref{con:value-load}, the joint state on $S \otimes B' \otimes B_V$ (conditional on $i$) is $\sum_j \widetilde a_{ij}\,|j\rangle_S \sum_m \sqrt{x_{jm}}\,|m\rangle_{B'} |w_m\rangle_{B_V}$. The $S$- and $B'$-register basis states are orthonormal, so tracing them out eliminates all $j \neq j'$ and $m \neq m'$ cross terms, leaving the displayed $\rho_i$. Its computational-basis diagonal is $\langle k|\rho_i|k\rangle = \sum_j |\widetilde a_{ij}|^2 \sum_m x_{jm}\, |\langle k|w_m\rangle|^2 = \sum_j A_{ij} \sum_m x_{jm} (W_V)_{km} = \sum_j A_{ij} (W_V x_j)_k$, using $|\langle k|w_m\rangle|^2 = (W_V)_{km}$, which holds because each column of $W_V$ lies on $\Delta^{d-1}$ (Assumption~\ref{assn:simplex}) so that $|w_m\rangle = \sum_k \sqrt{(W_V)_{km}}\,|k\rangle$ is a valid unit state. The channel involves no post-selection: the diagonal is exact at unit success probability.
\end{proof}

The channel-diagonal identity at the heart of the proof is Lean-verified in Appendix~\ref{sec:leancode}.

\section{The gated single-ancilla LCU residual}
\label{sec:residual}

The last classical operation is the weighted summation $\tfrac12(Y_i+x_i)$. A linear combination of unitaries (LCU) realizes this sum with one ancilla. Prepare the ancilla at an angle $\eta$, apply the sublayer conditioned on the ancilla, and combine the branches. The factor of $\tfrac12$ is achieved by setting $\eta=\pi/2$. Two readout options appear below. The first is the coherent readout (interfere and post-select the ancilla) which combines amplitudes and gives the general gated family from Theorem~\ref{thm:gated-residual}. The second is the incoherent readout (trace the ancilla out) which combines probabilities and reproduces Eq.~\ref{eq:classical} exactly (Remark~\ref{rem:incoherent-residual}). This circuit is shown in \Cref{fig:lcuresidual}. The incoherent readout is the one Theorem~\ref{thm:equiv} uses. The coherent family of Theorem~\ref{thm:gated-residual} is the amplitude-level generalization of that readout.

\begin{theorem}[Gated single-ancilla LCU residual]
\label{thm:gated-residual}
Let $U_{\mathrm{full}}$ denote the unitary cascade of Sections~\ref{sec:encoding}--\ref{sec:weighted-sum} up to (but excluding) the partial trace of Theorem~\ref{thm:weighted-sum}, followed by a terminal $\mathrm{SWAP}(B, B_V)$ that relocates the value register into the data register $B$ holding the input lift, acting on encoded $|X\rangle$ (with $S$, $B'$, and $B_V$ initialized to $|0\rangle$). The terminal SWAP is what makes the residual well-posed: without it, the identity branch of the LCU below carries its data in $B$ (with $B_V = |0\rangle$) while the sublayer branch carries its data in $B_V$, and the two branches would combine states living in different registers; with it, both branches present their output in $B$, and every residual readout below refers to $B$. Let $\eta$ be a learnable angle. The circuit
\begin{align}
\text{(i)}~&\text{prepare } |\eta\rangle_C = \cos(\eta/2)|0\rangle_C \nonumber\\
&\qquad\qquad + \sin(\eta/2)|1\rangle_C, \\
\text{(ii)}~&\text{controlled-}U_{\mathrm{full}}\ \text{conditioned on}\ |1\rangle_C, \\
\text{(iii)}~&\text{Hadamard on}\ C, \\
\text{(iv)}~&\text{post-select}\ |0\rangle_C,
\end{align}
produces
\begin{equation}
|y^{\mathrm{out}}\rangle \;\propto\; \cos(\eta/2)\,|X\rangle \;+\; \sin(\eta/2)\,U_{\mathrm{full}}|X\rangle,
\end{equation}
the parameterized residual sum with gate angle $\eta$. The success probability is
\begin{equation}
p_{\mathrm{succ}} = \tfrac12 \big(1 + \sin\eta\,\Real\langle X|U_{\mathrm{full}}|X\rangle\big).
\end{equation}
\end{theorem}
\begin{proof}
After step (i) the joint state is
\begin{equation*}
\big(\cos(\eta/2)|0\rangle_C + \sin(\eta/2)|1\rangle_C\big)\,|X\rangle.
\end{equation*}
After step (ii), with controlled-$U_{\mathrm{full}}$ conditioned on $|1\rangle_C$,
\begin{equation*}
\cos(\eta/2)|0\rangle_C|X\rangle + \sin(\eta/2)|1\rangle_C U_{\mathrm{full}}|X\rangle.
\end{equation*}
Step (iii) Hadamards the ancilla, mapping $\{|0\rangle, |1\rangle\} \to \{|+\rangle, |-\rangle\}$:
\begin{align*}
&\tfrac{1}{\sqrt 2}|0\rangle_C \big(\cos(\eta/2)|X\rangle + \sin(\eta/2) U_{\mathrm{full}}|X\rangle\big) \\
&\quad + \tfrac{1}{\sqrt 2}|1\rangle_C \big(\cos(\eta/2)|X\rangle - \sin(\eta/2) U_{\mathrm{full}}|X\rangle\big).
\end{align*}
Step (iv) projects onto $|0\rangle_C$, giving (after normalization) $\cos(\eta/2)|X\rangle + \sin(\eta/2) U_{\mathrm{full}}|X\rangle$. The success probability is the squared norm of the unnormalized $|0\rangle_C$-branch amplitude, $\tfrac12 \big\|\cos(\eta/2)|X\rangle + \sin(\eta/2) U_{\mathrm{full}}|X\rangle\big\|^2 = \tfrac12\big(1 + \sin\eta\,\Real\langle X|U_{\mathrm{full}}|X\rangle\big)$.
\end{proof}

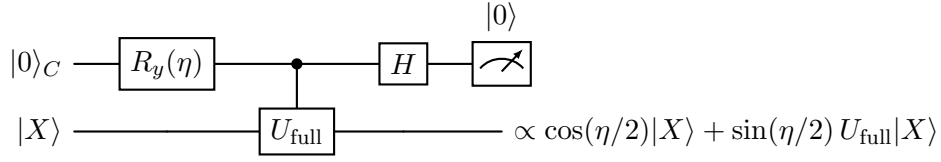
\begin{figure}
\begin{center}
\resizebox{\columnwidth}{!}{%
\begin{quantikz}[column sep=0.3cm, row sep=0.25cm]
\lstick{$|0\rangle_C$} & \gate{R_y(\eta)} & \ctrl{1} & \gate{H} & \meter{|0\rangle} \\
\lstick{$|X\rangle$} & & \gate{U_{\mathrm{full}}} & & \rstick{\shortstack[l]{$\propto \cos(\eta/2)|X\rangle$\\$\quad+\sin(\eta/2)\,U_{\mathrm{full}}|X\rangle$}}
\end{quantikz}}
\end{center}
\caption{Gated single-ancilla linear combination of unitaries residual addition diagram.}
\label{fig:lcuresidual}
\end{figure}

\subsection{Special cases of the gate angle}

Theorem~\ref{thm:gated-residual} gives a one-parameter family of outputs. Reading its endpoints and midpoint shows that the family contains three residual forms in common use, and that a learnable $\eta$ is a quantum analog of a soft gate.

\begin{corollary}[Gate-angle interpretations]
\label{cor:gate-interp}
The gate angle $\eta$ in Theorem~\ref{thm:gated-residual} produces the following residual modes:
\begin{itemize}[leftmargin=*]
\item $\eta = 0$: identity residual; the output equals the input $|X\rangle$. The sublayer is bypassed entirely.
\item $\eta = \pi$: full sublayer; the output equals $U_{\mathrm{full}}|X\rangle$. No residual.
\item $\eta = \pi/2$: equal-weight LCU residual; the output is the symmetric sum $\propto |X\rangle + U_{\mathrm{full}}|X\rangle$. This is the standard transformer additive identity.
\item $\eta$ learnable: a quantum-native gated residual, generalizing both the additive identity ($\eta = \pi/2$ fixed) and the convex-combination soft gate (any $\eta \in (0, \pi)$).
\end{itemize}
The classical-style soft gate $g \cdot U|X\rangle + (1-g)|X\rangle$ with $g \in [0,1]$ realizes the same family of output states (after $L^2$ renormalization) as the LCU output of Theorem~\ref{thm:gated-residual}, related by the bijective parameter correspondence
\begin{equation}
\begin{gathered}
\frac{g}{1-g} \;=\; \tan(\eta/2),\\
\text{equivalently}\quad
g \;=\; \frac{\sin(\eta/2)}{\sin(\eta/2) + \cos(\eta/2)}.
\end{gathered}
\end{equation}
The endpoints match the LCU special cases: $\eta = 0 \leftrightarrow g = 0$ (bypass), $\eta = \pi/2 \leftrightarrow g = 1/2$ (symmetric LCU), $\eta = \pi \leftrightarrow g = 1$ (full sublayer).
\end{corollary}
\begin{proof}
Direct substitution of the special values into the form of Theorem~\ref{thm:gated-residual}. For the parameter correspondence, two unit-norm states proportional to $\cos(\eta/2)|X\rangle + \sin(\eta/2)U|X\rangle$ and $(1-g)|X\rangle + g\,U|X\rangle$ coincide iff their coefficient ratios match, i.e., $\sin(\eta/2)/\cos(\eta/2) = g/(1-g)$.
\end{proof}

Corollary~\ref{cor:gate-interp} establishes that the soft-gated residual is realized entirely by the parameterized ancilla preparation as a unitary operation. The gate-fixed-at-$0.5$ LCU corresponds to $\eta = \pi/2$, the no-residual case to $\eta = 0$, and the full-sublayer case to $\eta = \pi$.

\begin{remark}[Incoherent variant recovers the probability-level residual exactly]
\label{rem:incoherent-residual}
The coherent LCU of Theorem~\ref{thm:gated-residual} interferes \emph{amplitudes}: at $\eta = \pi/2$ the Born readout of $\propto \sqrt{x} + \sqrt{y}$ (componentwise, for positive lifts) is in general \emph{not} the probability-level average $\tfrac12(x + y)$ of Eq.~\ref{eq:classical}; for example, $x = (1, 0)$ and $y = (\tfrac12, \tfrac12)$ give readout $(0.854, 0.146) \neq (0.75, 0.25)$. The probability-level residual is recovered exactly by the \emph{incoherent} variant of the same circuit: prepare the gating ancilla in $\cos(\eta/2)|0\rangle + \sin(\eta/2)|1\rangle$, apply controlled-$U_{\mathrm{full}}$, and \emph{trace out} the ancilla (no Hadamard, no post-selection). The joint state is then the mixture $\cos^2(\eta/2)\,|X\rangle\!\langle X| + \sin^2(\eta/2)\,U_{\mathrm{full}}|X\rangle\!\langle X|U_{\mathrm{full}}^{\dagger}$; reducing onto the data register $B$ (tracing out $S$, $B'$, and $B_V$, with the terminal SWAP of Theorem~\ref{thm:gated-residual} ensuring both branches present their data in $B$) gives a state whose computational-basis diagonal is exactly $\cos^2(\eta/2)\,x + \sin^2(\eta/2)\,y$; at $\eta = \pi/2$ this is $\tfrac12(x + y)$, with unit success probability. The equivalence theorem below uses the incoherent variant for its exactness claim; the coherent post-selected LCU is the strictly-larger amplitude-level generalization, coinciding with the incoherent readout for generic inputs only at the endpoints $\eta \in \{0, \pi\}$.
\end{remark}

\section{Equivalence theorem}
\label{sec:equiv}

The per-primitive correspondences of Sections~\ref{sec:encoding}--\ref{sec:residual} compose into the full single-head equivalence below.

\begin{theorem}[Quantum equivalent of single-head attention with residual]
\label{thm:equiv}
Adopt Assumption~\ref{assn:simplex} and the conventions of Remark~\ref{rem:conventions}, and assume the variational register has at least $\lceil \log_2 d\rceil + 1$ qubits. Permit one measure-and-reload step per attention scoring: the Hadamard-test outcome $z_{ij}$ is sampled, the row shift of Remark~\ref{rem:conventions} is applied, and the resulting rotation angle $\theta_{ij} = 2\arccos(e^{z_{ij}/2})$ is reloaded as the controlled-$R_y$ gate parameter of the softmax step.

For any input $X$ satisfying the assumption, there exists a quantum circuit $\mathcal{Q}$ and learnable parameters $\{\boldsymbol\Theta_Q, \boldsymbol\Theta_K, \boldsymbol\Theta_V, \eta\}$ (all rotation-gate angles) such that the Born-rule readout of $\mathcal{Q}(X)$ equals the classical single-head self-attention layer with residual $\mathcal{T}^{\mathrm{res}}(X)$ (Eq.~\ref{eq:classical}, read under the conventions of Remark~\ref{rem:conventions}) exactly in the infinite-shot limit under the measure-and-reload step, via the following primitive correspondences:
\begin{itemize}[leftmargin=*]
\item Encoding: Definition~\ref{def:encoding}.
\item Q, K projections: block-encoded $W_Q, W_K$ (Lemma~\ref{lem:embed}); the dilation unitary admits an exact decomposition into rotation gates and CNOTs by standard synthesis~\cite{shende2006synthesis,barenco1995elementary}, so the exactness of this correspondence does not rest on the density assumption. The rotation-CRY ansatz (Construction~\ref{con:rot-ansatz}, Lemma~\ref{lem:rot-ansatz-realization}) is the trainable finite-depth family for the same objects, conditional on Lemma~\ref{lem:rot-ansatz-realization}(3).
\item Softmax: $\cos^2$-Born readout (Proposition~\ref{prop:cossoftmax}) with $\theta_{ij}$ from the measure-and-reload step exactly reproduces exp-softmax (Theorem~\ref{thm:isomorphism}); Theorem~\ref{thm:trotter} gives an equivalent multi-stage realization.
\item Value aggregation: the column-stochastic $W_V$ realized exactly as a quantum channel via controlled column loading (Construction~\ref{con:value-load}), with the column state-preparation angles as the learnable parameters $\boldsymbol\Theta_V$; the diagonal of the partial trace equals the classical probability vector $Y_i = \sum_j A_{ij} W_V x_j$ exactly and deterministically (Theorem~\ref{thm:weighted-sum}).
\item Residual: the incoherent gated-ancilla mixture at $\eta = \pi/2$ reproduces the probability-level residual of Eq.~\ref{eq:classical} exactly (Remark~\ref{rem:incoherent-residual}); the coherent single-ancilla LCU (Theorem~\ref{thm:gated-residual}) is the amplitude-level generalization within the same circuit topology (Corollary~\ref{cor:gate-interp}).
\end{itemize}
\end{theorem}
\begin{proof}
Composition of the per-primitive correspondences listed. Each correspondence is established in its respective theorem; under the probability-space assumption and the qubit-count condition the composition realizes the classical attention layer's softmax pattern and value aggregation exactly and its additive residual exactly via the incoherent gating variant (Remark~\ref{rem:incoherent-residual}). The measure-and-reload step is the bridge from the Hadamard-test outcome $z_{ij}$ to the controlled-$R_y$ rotation angle $\theta_{ij}$ via the angle bijection of Theorem~\ref{thm:isomorphism}, after the row shift of Remark~\ref{rem:conventions}. The step is also the only place where the construction estimates anything from samples, and every other step is an exact identity, so the equivalence is exact in the infinite-shot limit.
\end{proof}

\paragraph{Measure-and-reload as a standard quantum-classical interface.}
The measure-and-reload step is a generic quantum-computing pattern: an intermediate ancilla is measured, a classical function is applied to the outcome, and the result is reloaded into the circuit as a downstream gate parameter. The same mechanism appears in adaptive quantum measurements~\cite{higgins2007entanglement,wiseman2010quantum}, quantum error correction syndrome extraction~\cite{fowler2012surface,terhal2015quantum}, magic-state distillation~\cite{bravyi2005universal}, and more generally in any algorithm where a piece of classical information extracted from one stage drives the parameter of a later stage. Its overhead of finite-shot statistical noise on the measurement, mid-circuit classical-control latency, and the requirement of mid-circuit measurement capability on the hardware is separately optimizable in principle: amplitude estimation~\cite{brassard2002quantum} gives $\sqrt{N}$ shot scaling, low-latency classical-control electronics reduce dispatch time, and mid-circuit measurement is a routine primitive on contemporary superconducting platforms~\cite{corcoles2021exploiting, govia2023randomized}. The equivalence of Theorem~\ref{thm:equiv} is therefore exact (in the infinite-shot limit), with the practical accuracy bounded by the shot budget allocated to each measure-and-reload step.

An expanded step-by-step audit of the construction's classical-side operations (initial state preparation, gate parameter loading at compile time, per-scoring measure-and-reload, and terminal Born-rule readout) appears in Section~\ref{sec:quantumness}.

At the macro level, the full circuit is the LCU residual sandwich applied to the unitary cascade $U_{\mathrm{full}}$ that realizes encoding, the $\cos^2$-softmax (controlled-$R_y$ on the attention register $S$ with angles supplied by the measure-and-reload step), and the value channel (controlled column loading onto $B_V$) (see \Cref{fig:completecircuit}).
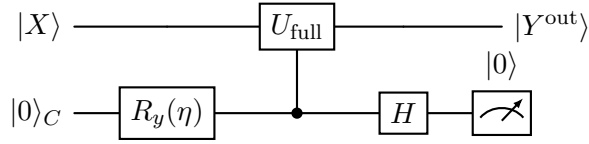
\begin{figure}[]
\begin{center}
\begin{quantikz}[column sep=0.6cm]
\lstick{$|X\rangle$} & & \gate{U_{\mathrm{full}}} & & \rstick{$|Y^{\mathrm{out}}\rangle$} \\
\lstick{$|0\rangle_C$} & \gate{R_y(\eta)} & \ctrl{-1} & \gate{H} & \meter{|0\rangle}
\end{quantikz}
\end{center}
\caption{Complete single-head quantum attention and residual addition diagram.}
\label{fig:completecircuit}
\end{figure}

The unitary $U_{\mathrm{full}}$ expands as $V_{x_i} \to U_{\mathrm{attn}}(\boldsymbol\theta_i) \to U_{\mathrm{enc}} \to U_V \to \mathrm{SWAP}(B, B_V)$ across the diagrams of Sections~\ref{sec:encoding}--\ref{sec:weighted-sum}, with the partial traces over $S$ and $B'$ deferred to readout. The measure-and-reload step sets $\boldsymbol\theta_i = 2\arccos(e^{\boldsymbol z_i/2})$ from the Hadamard-test outcomes; the LCU at $\eta = \pi/2$ realizes the additive residual.

\section{Formal verification in Lean 4}
\label{sec:lean}

The Lean code listings included in Appendix~\ref{sec:leancode} are the actual machine-checked proofs accepted by the Lean 4 kernel against mathlib. The full library compiles successfully with \texttt{lake build} (zero errors).

\paragraph{Scope of verification.}
The verified primitive identities are (i) the angle-score bijection on the simplex interior (Theorem~\ref{thm:isomorphism}), (ii) the boundary strict-extension claim with explicit witness (Theorem~\ref{thm:strict-extension}), (iii) the multi-stage post-selection algebraic identity for any positive weights, including the inverse-temperature generalization and the repeated-measurement special case (Theorem~\ref{thm:trotter}, Corollary~\ref{cor:temperature}), (iv) the gated single-ancilla circuit producing the parameterized residual sum (Theorem~\ref{thm:gated-residual}), (v) the identity-at-origin property of the rotation-CRY block-encoding ansatz (Lemma~\ref{lem:rot-ansatz-realization}, part 2), (vi) the polynomial-exponential impossibility at the algebraic core of the coherent no-go (Theorem~\ref{thm:no-go}), (vii) the channel-diagonal identity of the stochastic value channel (Theorem~\ref{thm:weighted-sum}). The composition into the full single-head equivalence at the algebraic level is also Lean-formalized, as \texttt{full\_single\_head\_equivalence} in \texttt{Master.lean}: cos$^2$-softmax under the angle bijection, composed with the weighted-sum value aggregation and the convex residual, equals the classical exp-softmax attention layer pointwise; the channel-instantiated form \texttt{full\_single\_head\_equivalence\_channel} carries the same equivalence end-to-end from the Born amplitudes $\sqrt{(W_V)_{km}}$ of the value channel. Each is proved from first principles using only mathlib primitives such as \texttt{Real.exp}, \texttt{Real.cos}, \texttt{Real.arccos}, \texttt{Real.cos\_arccos}, \texttt{Finset.sum}, and \texttt{Complex.normSq}.

\paragraph{Imported as standard results.}
The construction relies on standard quantum-computing infrastructure not currently formalized in mathlib: data-reuploading universal-approximation results~\cite{schuld2021effect}, the Solovay--Kitaev theorem, and the QSVT framework underlying the $\epsilon$-approximation theorem. These are cited as established results, qualified by the architectural assumptions of their respective theorems.

\paragraph{Out of scope.}
The architectural-level account of the classical-side interface operations (Section~\ref{sec:quantumness}) is an expository description, not a theorem. Density of the rotation-CRY ansatz in $SU(2^{b+1})$ (Lemma~\ref{lem:rot-ansatz-realization}, part 3) and the QSVT $\varepsilon$-approximation (Theorem~\ref{thm:qsvt}) are cited from the literature. The amplitude-level gate-sequence realization of the residual (Theorem~\ref{thm:gated-residual}) is Lean-checked at the LCU sandwich identity but not composed gate-by-gate with the value-aggregation circuit; the algebraic equivalence above takes the convex residual as given at the probability level, in agreement with Eq.~\ref{eq:classical} and realized exactly by the incoherent gating variant (Remark~\ref{rem:incoherent-residual}), with the coherent amplitude-level analog given by the LCU at $\eta = \pi/2$.

\paragraph{Reproducing the build.}
The Lean project is a standard Lake build pinned to the Lean 4 toolchain at version \texttt{v4.30.0-rc2}, with \texttt{mathlib4} as a dependency (auto-fetched via Lake). Running \texttt{lake build} yields a successful kernel check, which serves as a type-theoretic certificate that every theorem in Appendix~\ref{sec:leancode} has the type asserted by its signature; the Lean kernel rejects any step not reducible to mathlib axioms.

\paragraph{Code availability.}
The Lean 4 code is provided in Appendix~\ref{sec:leancode}.

\section{Quantum vs. classical operations}
\label{sec:quantumness}

Every step of the construction (Theorem~\ref{thm:equiv}) is either a unitary gate operation, a coherent ancilla-mediated transformation, or one of four standard quantum-classical interface operations: initial state preparation, gate parameter loading at compile time, the measure-and-reload step that supplies the softmax angle $\theta_{ij}$ from the Hadamard-test outcome $z_{ij}$, and terminal Born-rule readout.

\paragraph{What the quantum circuit does, step by step.}
\begin{enumerate}[leftmargin=*]
\item \textbf{Encoding.} Unitary state preparation $V_X: |0\rangle^{\otimes (a+b)} \mapsto |X\rangle_{AB}$; the classical input $X$ is loaded once at the start.
\item \textbf{Q, K block encoding.} Each score projection $W_*(\boldsymbol\Theta_*)$, $* \in \{Q, K\}$, is realized as a depth-$L$ rotation-CRY ansatz on $b+1$ qubits (Construction~\ref{con:rot-ansatz}) with the block-encoding ancilla post-selected on $|0\rangle$. The learned angles $\boldsymbol\Theta_*$ are gate parameters, not classical features; the ancilla is recycled across Q and K via reset between projections.
\item \textbf{Pairwise overlap.} The Hadamard test (Lemma~\ref{lem:overlap}) on a second ancilla produces a joint state whose ancilla $|0\rangle$-branch probability equals $(1 + \Real\langle q_i \mid k_j\rangle)/2$. Sampling this ancilla yields $z_{ij}$ (the measure step of measure-and-reload).
\item \textbf{Classical conversion and reload.} A classical control system computes $\theta_{ij} = 2\arccos(e^{z_{ij}/2})$ from each sampled $z_{ij}$ (after the row-max shift of Remark~\ref{rem:conventions}, so all $n$ scores of row $i$ are collected first) and reloads $\theta_{ij}$ as the rotation angle of the controlled-$R_y$ softmax gate (the reload step). The fixed bijection (Theorem~\ref{thm:isomorphism}) carries no learned parameters; the conversion is a deterministic arithmetic step on the classical control system, separate from any optimizer.
\item \textbf{Softmax.} Controlled-$R_y(\theta_{ij})$ on the attention register $S$ realizes $\cos^2$-Born readout (Proposition~\ref{prop:cossoftmax}) which under the reloaded angles reproduces exp-softmax exactly (Theorem~\ref{thm:isomorphism}).
\item \textbf{Value channel and aggregation.} The $j$-controlled token encoding and the column-loading channel $U_V$ (Construction~\ref{con:value-load}), followed by partial trace over $S$ and $B'$, give the classical $Y_i$ on the diagonal of the reduced density matrix (Theorem~\ref{thm:weighted-sum}). The channel is deterministic (no post-selection), and its learned angles $\boldsymbol\Theta_V$ are input-independent compile-time parameters.
\item \textbf{Gated residual.} Gated single-ancilla LCU at $\eta = \pi/2$ (Theorem~\ref{thm:gated-residual}) realizes the additive residual at the amplitude level.
\item \textbf{Readout.} A single round of Born-rule measurement of the data register returns the output probability vector.
\end{enumerate}

\paragraph{Classical-side operations.} The construction uses four classical-side operations, all of which are standard quantum-computing interface primitives:
\begin{description}[leftmargin=*]
\item[\textbf{Initial state preparation.}] The classical input $X$ is loaded once via a unitary state-prep circuit. This is the standard amplitude-encoding interface that many quantum machine learning algorithms use to ingest classical data.
\item[\textbf{Gate parameter loading at compile time.}] The learned rotation angles $\boldsymbol\Theta$ and any fixed circuit parameters are stored on the classical control system and deposited into the rotation-gate parameter fields at circuit-compilation time. This is identical to how common variational quantum algorithms (VQE, QAOA, QCBM) treat their learned parameters; the parameters are circuit constants set once per shot.
\item[\textbf{Measure-and-reload (per scoring pair).}] The Hadamard-test ancilla is sampled, the classical function $\theta = 2\arccos(e^{z/2})$ is computed on the control system, and $\theta$ is reloaded as a downstream controlled-$R_y$ angle. Mid-circuit measurement is a routine primitive on contemporary superconducting hardware~\cite{corcoles2021exploiting, govia2023randomized}. The overhead from finite-shot statistical noise on $z$, classical-control latency, and reload bandwidth is in principle separately optimizable (amplitude estimation~\cite{brassard2002quantum} gives $\sqrt N$ shot scaling; low-latency control electronics shorten the reload cycle).
\item[\textbf{Final measurement.}] A single Born-rule measurement at the end of the circuit returns the output probability vector. No subsequent quantum operation follows.
\end{description}

\paragraph{Coherence caveat.}
The construction of Theorem~\ref{thm:equiv} is exact but not fully coherent. The measure-and-reload step supplies the softmax angle $\theta_{ij}=2\arccos(e^{z_{ij}/2})$ from the classical outcome of the Hadamard test. Removing this step and demanding a fully-coherent finite-depth realization of $z\mapsto\exp(z)/\sum_l\exp(z_{il})$ from quantum-loaded scores runs into Theorem~\ref{thm:no-go}. There are two regimes which escape the no-go depending on the number of tokens $n$:
\begin{itemize}
    \item \textbf{$n=1$:} the softmax is trivially uniform on a single element and no coherent computation is required. This is the degenerate case where ``fully quantum'' holds, but it is not the setting of practical attention.
    \item \textbf{$n\geq2$:} Theorem~\ref{thm:qsvt} gives a fully-coherent $\varepsilon$-approximate route via QSVT polynomial dressing, at cost $O(\sqrt{\max\{c, \log(1/\varepsilon)\}\,\log(1/\varepsilon)})$ in additional circuit depth for scores bounded by $c$.
\end{itemize}
The measure-and-reload construction is what enables representation of the non-trivial $n\geq2$ cases.

\section{Practicality and scaling}\label{sec:practicality}

The equivalence theorem (Theorem~\ref{thm:equiv}) is stated in the infinite-shot limit and in gate-count-agnostic primitives (block-encoded projections, Hadamard tests, controlled state preparation). This section audits the near-term cost of that construction with focus on three critical scaling considerations: state-preparation and oracle model, success probability of the post-selection steps and whether amplitude amplification is needed, and heuristic finite-shot error propagation into the reload angle and softmax weights. We make no claims of a hardware roadmap but instead provide heuristic estimates of resource scaling.

\subsection{State-preparation and oracle model}
The construction takes as an oracle a state-preparation unitary $B_{x_i}$ for each classical input token $x_i \in \Delta^{d-1}$. The cost of realizing $B_{x_i}$ depends on which oracle model is used. For an arbitrary $d$-dimensional probability vector with no exploitable structure, preparing $|x_i\rangle=\sum_j\sqrt{x_{i,j}}\,|j\rangle$ from $|0\rangle^{\otimes b}$, $b=\lceil\log_2 d\rceil$ requires $O(d)$ two-qubit gates in the worst case \cite{shende2006synthesis}. Repeated over $n$ tokens, one preparation of the encoded input costs $O(nd)$ gates; because preparation is repeated on every shot, the per-layer state-preparation cost is $O(nd)$ multiplied by the shot budget of the finite-shot analysis below (see the scaling summary at the end of this section). Under a QRAM oracle~\cite{giovannetti2008quantum}, a coherent superposition query costs $O(\log d)$ time after an $O(nd)$-gate one-time preprocessing. QRAM at meaningful scale is not currently available but is a common assumption when considering practical scaling of quantum algorithms for future systems~\cite{cherrat2024qvit, guo2024quantumtransformer}. If the classical input $X$ is itself the output of an upstream quantum computation (a preceding attention layer, a quantum feature map, or a QSVT projection), state preparation is free by construction and the oracle model reduces to standard unitary composition.

The block-encoded projections $W_Q, W_K$ are realized by depth-$L$ rotation-CRY ans\"atze on $b+1$ qubits (Construction~\ref{con:rot-ansatz}). Each has $O(Lb)=O(L \log d)$ gates. Under the density assumption of Lemma~\ref{lem:rot-ansatz-realization}(3), $L=\operatorname{poly}(d)$ is sufficient for a target Frobenius error on $W_*$. The exact polynomial degree is unknown. The value channel $U_V$ (Construction~\ref{con:value-load}) is an input-independent multiplexed state preparation, $O(d^2)$ two-qubit gates in the worst case~\cite{shende2006synthesis}, compiled once per trained model. The single-ancilla LCU residual and the Hadamard test are $O(1)$-gate primitives on top of the block encodings they wrap.

\subsection{Success probability and amplitude amplification}

Three stages of the construction post-select an ancilla. Those are the two block encodings, the softmax stage (using the post-selected approach), and the residual LCU (using the coherent approach). This subsection bounds the success probability and identifies where amplitude amplification might be needed.

\paragraph{Block-encoding post-selection.} Each of $W_Q, W_K$ uses one ancilla with block-encoding post-selection; the value channel involves no post-selection at all (Theorem~\ref{thm:weighted-sum}). On a single token, the ancilla-$|0\rangle$ success probability is $p_i=\|W_*|x_i\rangle\|^2=\langle x_i| W_*^\dagger W_*|x_i\rangle$. The average per-token success probability is then $\bar{p}=\operatorname{tr}(W_*^{\dagger}W_*\rho_B)$, where $\rho_B = \tfrac1n\sum_i |x_i\rangle\langle x_i|$, bounded by $\|W_*\|_{\mathrm{op}}^2$, the square of the largest singular value of $W_*$. For projections into the unit ball this is $O(1)$ for generic $|x\rangle$. Amplitude amplification boosts a lower bound success probability $p=\Omega(1)$ to near-unity in $O(1/\sqrt p)=O(1)$ upper bound of Grover iterations~\cite{brassard2002quantum}. Fixed-point amplitude amplification~\cite{yoder2014fixedpoint} gives the same asymptotic scaling with less fragility compared to the original construction.

\paragraph{Softmax post-selection.} Two routes are available for the exp-softmax:
\begin{itemize}
    \item \textbf{Single-stage exact route (Theorem~\ref{thm:isomorphism}).} Because the reloaded angles are classical, the normalized attention state is compiled directly as an $i$-controlled state preparation (Proposition~\ref{prop:cossoftmax}), and the $\cos^2$-Born readout returns $j$ with the correct exp-softmax marginal in a single measurement with no rejection; the cost lives entirely in the shot budget of the preceding Hadamard test. The literal controlled-$R_y$ variant with ancilla post-selection (Remark~\ref{rem:trotter-special-cases}) instead succeeds with probability $Z_i/n \geq 1/n$, requiring $O(\sqrt{n})$ amplitude-amplification iterations in the worst case.
    \item \textbf{$L$-stage post-selected route (Theorem~\ref{thm:trotter}).} Per-pair success $e^{z_{ij}}$, which is exponentially small in $|z_{ij}|$ for heavily-suppressed token pairs. Per-pair amplitude amplification gives $O(e^{|z_{ij}|/2})$ Grover iterations, the standard quadratic speedup over the $O(e^{|z_{ij}|})$ shots required at success probability $e^{z_{ij}}$.
\end{itemize}

\paragraph{Residual LCU post-selection.} The single-ancilla LCU (Theorem~\ref{thm:gated-residual}) has ancilla-$|0\rangle$ success probability $p_{\mathrm{succ}} = \tfrac12\big(1 + \sin\eta\,\Real\langle X|U_{\mathrm{full}}|X\rangle\big)$. For a generic unitary this can vanish: at $\eta = \pi/2$ with $U_{\mathrm{full}}|X\rangle = -|X\rangle$, the additive-residual point has success probability zero. The simplex setting rules this case out: every computational-basis amplitude of both $|X\rangle$ and $U_{\mathrm{full}}|X\rangle$ is nonnegative (positive-orthant lifts; $\cos(\theta/2), \sin(\theta/2) \geq 0$ for $\theta \in [0, \pi]$ acting on registers initialized in $|0\rangle$; value amplitudes $\sqrt{V_{jk}} \geq 0$; and post-selection only renormalizes a nonnegative vector), so $\Real\langle X|U_{\mathrm{full}}|X\rangle \geq 0$ and $p_{\mathrm{succ}} \geq \tfrac12$ for every $\eta$. Amplitude amplification is not needed.

\subsection{Heuristic finite-shot error}
The measure-and-reload step samples the ancilla marginal $P_{ij}(|0\rangle)=(1+\Real\langle q_i|k_j\rangle)/2=(1+z_{ij})/2$ from a Bernoulli distribution. Here we define N as shots per row. A shot-averaged estimator over $N$ shots has variance $\sigma^2_{\widehat{z}}=(1-z^2)/N\leq1/N$ on the recovered score $\widehat{z}_{ij}$. Amplitude estimation gives $\sigma_{\widehat{z}}=O(1/N)$ instead of $O(1/\sqrt{N})$ at the cost of an $O(N)$-deep phase-estimation dressing~\cite{brassard2002quantum}.

Because all $n$ scores of a row are collected before any reload and the row shift of Remark~\ref{rem:conventions} subtracts the \emph{empirical} maximum, every shifted score satisfies $\widehat{z}_{ij} - \max_l \widehat{z}_{il} \leq 0$ by construction: the reload map $\theta = 2\arccos(e^{z/2})$ is always defined, and no clipping rule is needed. The shift itself introduces no bias: softmax is exactly invariant under any common shift of a row, including a random one, so the composite readout is exactly the exp-softmax of the noisy scores, $e^{\widehat{z}_{ij}}/\sum_l e^{\widehat{z}_{il}}$.

The Jacobian of the reload map on $z\in(-\infty,0]$, \[\frac{\partial\theta}{\partial z} \;=\; -\frac{e^{z/2}}{\sin(\theta/2)} \;=\;-\frac{e^{z/2}}{\sqrt{1-e^{z}}},\] diverges as $z \to 0^{-}$, but the divergence is a coordinate artifact of angle space, not an amplification of shot noise. The weight-space sensitivity is the composition $\partial A/\partial z = (\partial A/\partial \theta)(\partial\theta/\partial z)$ with $A = \cos^2(\theta/2)$: near $z = 0^-$ one has $\theta \approx 2\sqrt{-z}$ and $\partial A/\partial\theta = -\sin(\theta)/2 \approx -\sqrt{-z}$, so the two factors cancel and $\partial A/\partial z = e^{z} \leq 1$ is uniformly bounded on the whole domain. The same cancellation protects against control noise on the angle itself: $\partial A/\partial\theta \to 0$ exactly where $\partial\theta/\partial z$ diverges. Heavily-suppressed pairs ($z \ll 0$) are further damped by the factor $e^{z}$. The only bias that survives normalization is the Jensen bias of exponentiation, $\mathbb{E}[e^{\widehat{z}}] = e^{z}\,e^{\sigma^2_{\widehat{z}}/2}$: its row-common part cancels between numerator and denominator to first order, leaving a differential contribution of order $(\sigma^2_{\widehat{z}_{ij}} - \bar{\sigma}^2_{i})/2 = O(1/N)$ from the score dependence of the Bernoulli variance, subdominant to the $O(\sqrt{n/N})$ row-TV fluctuation computed next.

Downstream, the softmax weight $A_{ij}=\cos^2(\widehat{\theta}_{ij}/2)=e^{\widehat{z}_{ij}}$ has delta-method variance $\operatorname{Var}(A_{ij})\approx A_{ij}^2 \cdot \operatorname{Var}(\widehat z_{ij})=A_{ij}^2(1-z_{ij}^2)/N \le A_{ij}^2/N$. Summing over a row of length $n$, the total-variation error on the attention row is \[\mathbb{E}\|\widehat A_i - A_i\|_{\mathrm{TV}} \; \lesssim\;\sqrt{n/N}\] so achieving row-Total-Variation (row-TV) error $\varepsilon$ requires $N=\Omega(n/\varepsilon^2)$ shots per row under naive Bernoulli sampling, or $N=\Omega(n/\varepsilon)$ under amplitude estimation. Across all $n$ rows, the total measurement budget per layer is $\Omega(n^2/\varepsilon^2)$ shots (naive) or $\Omega(n^2/\varepsilon)$ with amplitude estimation. The residual gate angle $\eta$ is a compile-time parameter with no per-inference measurement and contributes no shot-noise term.

\paragraph{Order-of-magnitude example.} For inference at target attention-weight precision $\varepsilon=10^{-2}$ on an $n=64$-token sequence, naive sampling needs $\sim6\times10^{5}$ shots per row and $\sim4\times10^{7}$ shots per attention pass. Amplitude estimation reduces this to $\sim6\times10^{3}$ shots per row and $\sim 4 \times 10^{5}$ per pass at the cost of $\sim 10^{2}$-deeper phase-estimation dressing. These budgets are the practicality bottleneck for an $\ell$-layer transformer where the total budget accumulates as $\ell \times$ (per-layer).

\paragraph{Scaling summary.} For inference on classical data, the honest per-layer accounting multiplies gates by shots: each shot re-prepares the input ($O(nd)$ gates without QRAM), and the shot budget is $\Omega(n^2/\varepsilon^2)$ (naive) or $\Omega(n^2/\varepsilon)$ (amplitude estimation), for a per-layer total of $\Omega(n^3 d/\varepsilon^2)$ (respectively $\Omega(n^3 d/\varepsilon)$) gate applications against the classical $O(n^2 d)$ FLOPs. There is no cost regime in which this construction outperforms classical evaluation for single-layer inference on classical data; the contribution of Theorem~\ref{thm:equiv} is structural (an exactness statement), not a speedup claim. The accounting changes category when the input is quantum-native (produced by an upstream quantum layer, feature map, or QSVT projection), since classical evaluation is then unavailable without tomography; in that regime the layer adds only polynomial overhead on top of the pipeline that produced its input, and the relevant comparison is to other quantum layers, not to classical attention.

\section{Discussion}
\label{sec:discussion}

Classical attention's outputs are generally unbounded and, by definition, that regime cannot be recovered by a purely quantum ansatz. For models whose outputs are bounded or renormalized to the probability simplex, however, the result of this paper is stronger than representability: every mechanism of the classical single-head layer has a specific quantum analog that recovers it exactly, and the analogs compose into the full layer. Table~\ref{tab:dictionary} collects the dictionary established in Sections~\ref{sec:encoding}--\ref{sec:residual}, entry by entry, with the exactness status and the certifying result for each. Three entries are not merely recoveries but strict extensions or physical reinterpretations of the classical mechanism: the cosine-squared softmax reaches boundary distributions that exponential softmax cannot reach at finite scores (Theorem~\ref{thm:strict-extension}); the gated ancilla interpolates a one-parameter family of residual modes of which the classical additive identity is the $\eta = \pi/2$ slice (Corollary~\ref{cor:gate-interp}); and the temperature axis, a hyperparameter classically, is realized physically as post-selected measurement weight (Corollary~\ref{cor:temperature}).

\begin{table*}
\begin{center}
\small
\begin{tabular}{|p{0.2\textwidth} | >{\raggedright}p{0.25\textwidth} | >{\raggedright}p{0.3\textwidth} | >{\raggedright\arraybackslash}p{0.12\textwidth}|}
\toprule
classical mechanism & quantum analog & status & result \\
\midrule
simplex input & Born-rule amplitude lift & exact & Def.~\ref{def:encoding} \\
projections & block-encoded contractions & exact$^{\dagger}$ & Lem.~\ref{lem:embed}, \ref{lem:rot-ansatz-realization} \\
attention score & Hadamard-test ancilla marginal & exact (infinite shot) & Lem.~\ref{lem:overlap} \\
exp. softmax & Born-rule readout of amplitudes & exact; strictly extended & Thm.~\ref{thm:isomorphism}, \ref{thm:strict-extension} \\
temperature & repeated post-selected measurement & exact (post-selected) & Thm.~\ref{thm:trotter}, Cor.~\ref{cor:temperature} \\
value application & controlled column-loading channel & exact, deterministic & Con.~\ref{con:value-load}, Thm.~\ref{thm:weighted-sum} \\
residual with gate & single-ancilla preparation angle & exact (incoherent); amplitude-level generalization (coherent) & Thm.~\ref{thm:gated-residual}, Rem.~\ref{rem:incoherent-residual} \\
sparse/hard attention & boundary angles $\theta = \pi$ & strict extension & Thm.~\ref{thm:strict-extension} \\
\bottomrule
\end{tabular}
\end{center}
\caption{The quantum-classical dictionary for single-head attention on the probability simplex. $^{\dagger}$Existence of the block encoding is exact (Lemma~\ref{lem:embed}); its variational realization is conditional on the gate-set density assumption of Lemma~\ref{lem:rot-ansatz-realization}(3). Scores are read in the normalized amplitude-lift (Hellinger) convention of Remark~\ref{rem:conventions}.}
\label{tab:dictionary}
\end{table*}

\paragraph{Softmax as a slice of the simplex.}
Classical softmax projects a vector of unconstrained real scores onto the probability simplex via the exponential map, which requires unbounded scores to express extreme cases and never produces exact zeros at finite parameters. The cosine-squared softmax obtained from Born-rule measurement reaches the same interior distributions and additionally includes the boundary, where attention weights are exactly zero at finite parameter values. Existing sparse-attention variants such as sparsemax~\cite{martins2016softmax,correia2019adaptively} and entmax~\cite{peters2019sparse} achieve hard zeros through projection or $\alpha$-entmax thresholding; the cosine-squared parameterization achieves the same boundary expressivity through a smooth function with no exponentials or underflow concerns. The temperature axis itself is realized physically in this construction: the number of repeated post-selected Born measurements is exactly the inverse temperature of the attention distribution (Corollary~\ref{cor:temperature}), and hard attention sits at a finite angle rather than at the zero-temperature limit of an anneal. Whether the cosine-squared form is competitive with these classical sparse-attention variants in practice is an empirical question the equivalence theorem motivates, independent of any quantum implementation.

\paragraph{Ancilla as a tool to replace classical operations.}
Residual connections in deep learning have taken several forms: highway gates~\cite{srivastava2015highway}, ResNet's pure additive identity~\cite{he2016deep}, gated linear units~\cite{dauphin2017language}, and various softer or learnable variants. The single-ancilla LCU construction maps each of these to a specific ancilla preparation: $|0\rangle$ for bypass, $|+\rangle$ for the additive identity, $|1\rangle$ for the full sublayer, and a parameterized superposition $\cos(\eta/2)|0\rangle + \sin(\eta/2)|1\rangle$ for a learnable gate. All four are realized by the same circuit topology, namely one ancilla, controlled-$U$, Hadamard, and post-select, differing only in the preparation angle. The unification places these four residual variants in a single one-parameter family. The ancilla also functions to carry out query and key block encoding, attention score inner product calculation using the Hadamard test, and Trotter decomposition to improve fidelity in single-gate noise which can increase with large angle rotations. All of these operations can be carried out by recycling the same ancilla with measurement or reset between algorithm steps.

\paragraph{Implications and open directions.}
For empirical work on quantum hardware, the construction identifies the operational primitives required on hardware: block encoding for the projections, the Hadamard test for inner products, and the gated single-ancilla LCU for the residual. These are all well-studied operations, but their composition into a transformer attention layer has not previously been laid out as an exact construction with verification. Trainability under barren-plateau constraints~\cite{mcclean2018barren,cerezo2021cost} is a real concern that the formal proof does not resolve; the proof says nothing about whether the variational parameters are easy to learn, only that the architecture has the expressivity to match its classical counterpart in principle.

A separate direction is on the classical side. The cosine-squared softmax with free angle parameterization is strictly more expressive than exponential softmax at finite parameter values, because it reaches boundary points of the simplex with exact zero attention weights. To our knowledge, cosine-squared softmax has not been studied as a drop-in attention parameterization. This suggests a separate classical direction: testing whether the same boundary expressivity that arises naturally from Born-rule normalization provides practical benefits in simplex-valued attention models. The probability simplex setting is the natural fit: classical generative models that operate directly on the simplex  (Riemannian flow matching~\cite{chen2024riemannian}, Dirichlet flow matching~\cite{avdeyev2024dirichlet}, and continuous diffusion for categorical data~\cite{dieleman2022continuous}) would be the most direct testbed for cos\textsuperscript{2}-softmax as a drop-in replacement for the standard exponential softmax in their attention layers. No quantum advantage is claimed; whether the cosine-squared form yields measurable improvements in classical training is an empirical question outside the scope of this paper.

\paragraph{Multi-head extension.}
The construction is single-head. Extension to multi-head attention is the standard parallel composition: $H$ independent copies of the single-head circuit run with independent variational parameters $\{\boldsymbol\Theta^{(h)}\}_{h=1}^{H}$ and independent ancilla registers, with the per-head outputs combined classically at readout (matching how multi-head attention concatenates head outputs in the classical setting). The mathematical content of this paper transfers head-by-head; what does not transfer is the hardware-resource accounting, which scales linearly with $H$ in qubit count and gate count. Whether parameter sharing across heads (e.g., a single block-encoded $W_Q$ with head-indexed projection) admits a more economical realization is a natural follow-up.

\section{Conclusion}
\label{sec:conclusion}

We have constructed an exact equivalence, in the infinite-shot limit, between the classical single-head attention layer with residual connection on the probability simplex and a quantum circuit in which every learnable parameter is a rotation-gate angle. The equivalence is a dictionary rather than a simulation: scores, softmax, temperature, value aggregation, and the gated residual each correspond to a specific quantum primitive (\Cref{fig:headline}, Table~\ref{tab:dictionary}), with the softmax and residual families strictly extending their classical counterparts and the temperature axis realized physically as measurement repetition. The classical side of the interface reduces to state preparation, compile-time parameter loading, one measure-and-reload step per attention score, and a terminal Born-rule measurement. Exact fully-coherent realization at finite depth is provably impossible for affinely-encoded scores (Theorem~\ref{thm:no-go}), while QSVT dressing recovers it to accuracy $\varepsilon$ (Theorem~\ref{thm:qsvt}); the obstruction in both directions is the single non-polynomial reparameterization $z \mapsto e^{z/2}$, paid either classically or coherently. The algebraic core of the dictionary is machine-checked in Lean 4 against mathlib.

The probability simplex is the domain in which Riemannian and Dirichlet flow matching~\cite{chen2024riemannian,avdeyev2024dirichlet} and categorical diffusion~\cite{dieleman2022continuous} already operate, so the primitives identified here connect directly to an active classical research thread. In particular, the cosine-squared softmax is a strictly larger family than exponential softmax at finite parameter values and may merit study as a classical parameterization in its own right.

\section{AI Use Statement}
Claude Opus 4.6 was used for Lean code generation for proof verification. Claude Opus 4.8 was used for simulated journal review/criticism, spell-check, and grammar editing.

\bibliographystyle{unsrt}
\bibliography{references}

\appendix
\section{Lean Code}
\label{sec:leancode}

Verification of Lemma~\ref{lem:rot-ansatz-realization}, part (2): the constituent gates at zero parameters are identity operations on every qubit state.

\begin{lstlisting}
theorem applyRy_zero (x : Qubit) : applyRy 0 x = x := by
  simp [applyRy]

theorem applyRz_zero (x : Qubit) : applyRz 0 x = x := by
  simp [applyRz]

theorem applyRot_zero (x : Qubit) : applyRot 0 0 0 x = x := by
  unfold applyRot
  rw [applyRz_zero, applyRy_zero, applyRz_zero]

theorem applyCRy_zero (control : Fin 2) (target : Qubit) :
    applyCRy 0 control target = target := by
  unfold applyCRy
  by_cases h : control = 1
  all_goals simp [h, applyRy_zero]

theorem variational_origin_is_identity
    (control : Fin 2) (x : Qubit) :
    applyRot 0 0 0 (applyCRy 0 control x) = x := by
  rw [applyCRy_zero, applyRot_zero]
\end{lstlisting}

Verification of Theorem~\ref{thm:isomorphism} of $\cos^2$-softmax and exp-softmax interior isomorphism.

\begin{lstlisting}
theorem master_interior_isomorphism {n : ℕ} (z : Fin n → Fin n → ℝ)
    (hz : ∀ i j, z i j ≤ 0) :
    ∀ i j, cosSqSoftmax (fun a b => angleFromScore (z a b)) i j
         = expSoftmax z i j :=
  fun i j => cosSqSoftmax_eq_expSoftmax_under_bijection z hz i j
\end{lstlisting}

Verification of Theorem~\ref{thm:strict-extension} that $\cos^2$-softmax strictly extends the range of exponential softmax.

\begin{lstlisting}
theorem strict_extension (n : ℕ) (hn : 1 < n) :
    ∃ θ : Fin n → Fin n → ℝ, ∀ z : Fin n → Fin n → ℝ,
      ∃ i j : Fin n, cosSqSoftmax θ i j ≠ expSoftmax z i j := by
  obtain ⟨θ, i, j, h_zero⟩ := exists_cosSqSoftmax_zero n hn
  refine ⟨θ, fun z => ⟨i, j, ?_⟩⟩
  rw [h_zero]
  exact (expSoftmax_pos z i j).ne
\end{lstlisting}

Verification of Theorem~\ref{thm:trotter} that post-selecting $L$ stages whose weights sum to $\beta$ multiplies to a single exp-softmax at inverse temperature $\beta$. Two specific cases are the unit-temperature case of Eq.~\ref{eq:classical} and the equal-weight case of Corollary~\ref{cor:temperature}.

\begin{lstlisting}[escapeinside={(*}{*)}]
theorem trotter_factorization {L : ℕ} (α : Fin L → ℝ) (z : ℝ)
    (hα_sum : ∑ l, α l = 1) :
    ∏ l, exp (α l * z) = exp z := by
  rw [(*$\leftarrow$*) exp_sum]
  congr 1
  rw [(*$\leftarrow$*) sum_mul, hα_sum, one_mul]

theorem trotter_post_selection {n L : ℕ} (z : Fin n → Fin n → ℝ)
    (α : Fin L → ℝ) (hα_sum : ∑ l, α l = 1) (i j : Fin n) :
    ∏ l, exp (α l * z i j) = exp (z i j) :=
  trotter_factorization α (z i j) hα_sum

theorem trotter_temperature {L : ℕ} (α : Fin L → ℝ) (z β : ℝ)
    (hα_sum : ∑ l, α l = β) :
    ∏ l, exp (α l * z) = exp (β * z) := by
  rw [(*$\leftarrow$*) exp_sum]
  congr 1
  rw [(*$\leftarrow$*) sum_mul, hα_sum]

theorem trotter_repeated_measurement (L : ℕ) (z : ℝ) :
    ∏ _l : Fin L, exp z = exp ((L : ℝ) * z) := by
  have h := trotter_temperature (fun _ : Fin L => (1 : ℝ)) z L (by simp)
  simpa using h
\end{lstlisting}

Theorem~\ref{thm:no-go} algebraic core verification that a polynomial cannot agree with a genuine exponential on an open interval.

\begin{lstlisting}
lemma hasDerivAt_expAffine (C a b σ : ℝ) :
    HasDerivAt (fun t => C * Real.exp (a * t + b))
      (a * (C * Real.exp (a * σ + b))) σ := by
  have h1 : HasDerivAt (fun t : ℝ => a * t + b) a σ := by
    simpa using ((hasDerivAt_id σ).const_mul a).add_const b
  have h2 := ((Real.hasDerivAt_exp (a * σ + b)).comp σ h1).const_mul C
  convert h2 using 1
  ring

theorem no_polynomial_matches_exponential
    (Q : Polynomial ℝ) (C a b s t : ℝ)
    (hC : C ≠ 0) (ha : a ≠ 0) (hst : s < t)
    (hagree : ∀ σ ∈ Set.Ioo s t, Q.eval σ = C * Real.exp (a * σ + b)) :
    False := by
  obtain ⟨σ₀, hσ₀⟩ := Set.nonempty_Ioo.mpr hst
  have hQne : Q ≠ 0 := by
    intro h
    have h0 := hagree σ₀ hσ₀
    rw [h, Polynomial.eval_zero] at h0
    exact (mul_ne_zero hC (Real.exp_ne_zero _)) h0.symm
  have hderiv : ∀ σ ∈ Set.Ioo s t,
      Q.derivative.eval σ = a * Q.eval σ := by
    intro σ hσ
    have hnhds : (fun y => Q.eval y)
        =ᶠ[nhds σ] (fun y => C * Real.exp (a * y + b)) :=
      Filter.eventually_of_mem (isOpen_Ioo.mem_nhds hσ) hagree
    have heq := hnhds.deriv_eq
    rw [(Polynomial.hasDerivAt Q σ).deriv,
        (hasDerivAt_expAffine C a b σ).deriv] at heq
    rw [heq, hagree σ hσ]
  have hRroots : {x | (Q.derivative - Polynomial.C a * Q).IsRoot x}.Infinite := by
    apply (Set.Ioo_infinite hst).mono
    intro σ hσ
    show (Q.derivative - Polynomial.C a * Q).IsRoot σ
    simp only [Polynomial.IsRoot, Polynomial.eval_sub, Polynomial.eval_mul,
      Polynomial.eval_C]
    rw [hderiv σ hσ]
    ring
  have hR0 : Q.derivative - Polynomial.C a * Q = 0 :=
    Polynomial.eq_zero_of_infinite_isRoot _ hRroots
  have hQd : Q.derivative = Polynomial.C a * Q := sub_eq_zero.mp hR0
  have hlt : Q.derivative.degree < Q.degree :=
    Polynomial.degree_derivative_lt hQne
  rw [hQd, Polynomial.degree_C_mul ha] at hlt
  exact lt_irrefl _ hlt
\end{lstlisting}

Verification of the channel-diagonal identity of Theorem~\ref{thm:weighted-sum} that loading the square roots of the nonnegative channel weights as amplitudes and reading out probabilities returns the direct weighted sum of the input values.

\begin{lstlisting}
noncomputable def valueFromChannel {n d : ℕ} (W : Fin d → Fin d → ℝ)
    (x : Fin n → Fin d → ℝ) (j : Fin n) (k : Fin d) : ℝ :=
  ∑ m, W k m * x j m

theorem channel_diagonal {n d : ℕ} (W : Fin d → Fin d → ℝ)
    (hW : ∀ k m, 0 ≤ W k m) (x : Fin n → Fin d → ℝ)
    (j : Fin n) (k : Fin d) :
    ∑ m, x j m * (Real.sqrt (W k m)) ^ 2 = valueFromChannel W x j k := by
  unfold valueFromChannel
  apply Finset.sum_congr rfl
  intro m _
  rw [Real.sq_sqrt (hW k m), mul_comm]
\end{lstlisting}

Sanity check of correctness of LCU circuit from Theorem~\ref{thm:gated-residual}.

\begin{lstlisting}
theorem lcu_circuit_correct {N : ℕ} (η : ℝ) (U : Fin N → Fin N → ℂ)
    (X : State N) (i : Fin N) :
    lcuCircuit η U X i =
      ((1 : ℂ) / (Real.sqrt 2 : ℂ)) *
      ((Real.cos (η / 2) : ℂ) * X i +
       (Real.sin (η / 2) : ℂ) * applyU U X i) := by
  have h0 : ((⟨0, by omega⟩ : Fin 2).val = 0) = True := by simp
  have h1 : ((⟨1, by omega⟩ : Fin 2).val = 0) = False := by decide
  unfold lcuCircuit projectZero hadamardAncilla
  simp only [h0, if_true]
  unfold controlledU
  simp only [h0, h1, if_true, if_false]
  unfold initState
  simp only [h0, h1, if_true, if_false]
  unfold applyU
  rw [show ∑ j, U i j * ((Real.sin (η/2) : ℂ) * X j) =
         (Real.sin (η/2) : ℂ) * ∑ j, U i j * X j from by
    rw [Finset.mul_sum]; congr 1; funext j; ring]
\end{lstlisting}

Certification of three key claims of interior isomorphism (Theorem~\ref{thm:isomorphism}), strict extension (Theorem~\ref{thm:strict-extension}), and Trotter algebraic identity (Theorem~\ref{thm:trotter}).

\begin{lstlisting}[escapeinside={(*}{*)}]
theorem single_head_master_summary (n : ℕ) (hn : 1 < n)
    (z : Fin n → Fin n → ℝ) (hz : ∀ i j, z i j ≤ 0) :
    -- 1. Interior isomorphism:
    (∀ i j, cosSqSoftmax (fun a b => angleFromScore (z a b)) i j =
            expSoftmax z i j) ∧
    -- 2. Strict extension at boundary:
    (∃ θ : Fin n → Fin n → ℝ, ∀ z' : Fin n → Fin n → ℝ,
       ∃ i j : Fin n, cosSqSoftmax θ i j ≠ expSoftmax z' i j) ∧
    -- 3. Trotter algebraic identity for exp-softmax:
    (∀ (L : ℕ) (α : Fin L → ℝ) (i j : Fin n),
       (∑ l, α l = 1) → ∏ l, exp (α l * z i j) = exp (z i j)) := by
  refine ⟨?_, ?_, ?_⟩
  (*$\cdot$*) exact master_interior_isomorphism z hz
  (*$\cdot$*) exact strict_extension n hn
  (*$\cdot$*) intro L α i j hα
    exact trotter_post_selection z α hα i j
\end{lstlisting}

Algebraic level verification of Theorem~\ref{thm:equiv} which claims complete single-head attention equivalence.

\begin{lstlisting}
def classicalAttentionOutput {n d : ℕ}
    (z : Fin n → Fin n → ℝ) (V : Fin n → Fin d → ℝ)
    (x : Fin n → Fin d → ℝ) (i : Fin n) (k : Fin d) : ℝ :=
  ((∑ j, expSoftmax z i j * V j k) + x i k) / 2

def quantumAttentionOutput {n d : ℕ}
    (z : Fin n → Fin n → ℝ) (V : Fin n → Fin d → ℝ)
    (x : Fin n → Fin d → ℝ) (i : Fin n) (k : Fin d) : ℝ :=
  ((∑ j, cosSqSoftmax (fun a b => angleFromScore (z a b)) i j * V j k)
    + x i k) / 2

theorem full_single_head_equivalence {n d : ℕ}
    (z : Fin n → Fin n → ℝ) (hz : ∀ i j, z i j ≤ 0)
    (V : Fin n → Fin d → ℝ) (x : Fin n → Fin d → ℝ)
    (i : Fin n) (k : Fin d) :
    quantumAttentionOutput z V x i k = classicalAttentionOutput z V x i k := by
  unfold quantumAttentionOutput classicalAttentionOutput
  congr 1
  congr 1
  apply Finset.sum_congr rfl
  intro j _
  rw [master_interior_isomorphism z hz i j]

theorem full_single_head_equivalence_channel {n d : ℕ}
    (z : Fin n → Fin n → ℝ) (hz : ∀ i j, z i j ≤ 0)
    (W : Fin d → Fin d → ℝ) (hW : ∀ k m, 0 ≤ W k m)
    (x : Fin n → Fin d → ℝ) (i : Fin n) (k : Fin d) :
    ((∑ j, cosSqSoftmax (fun a b => angleFromScore (z a b)) i j *
        (∑ m, x j m * (Real.sqrt (W k m)) ^ 2)) + x i k) / 2
      = classicalAttentionOutput z (valueFromChannel W x) x i k := by
  have h : ∀ j : Fin n,
      (∑ m, x j m * (Real.sqrt (W k m)) ^ 2) = valueFromChannel W x j k :=
    fun j => channel_diagonal W hW x j k
  simp_rw [h]
  exact full_single_head_equivalence z hz (valueFromChannel W x) x i k
\end{lstlisting}

\end{document}